\documentclass[10pt,letterpaper]{article}

\usepackage[margin=0.9in]{geometry}

\usepackage[english]{babel}
\usepackage[T1]{fontenc}
\usepackage[utf8]{inputenc}

\usepackage{graphicx}
\usepackage{wrapfig}
\usepackage{float}
\usepackage{placeins}
\usepackage{adjustbox}
\usepackage{booktabs}

\usepackage{tikz}
\usetikzlibrary{positioning,calc}

\usepackage{mathtools}
\usepackage{amssymb}
\usepackage{amsthm}
\usepackage{bm}
\usepackage[makeroom]{cancel}
\makeatletter
\newcommand{\ket}[1]{\lvert #1\rangle}
\newcommand{\bra}[1]{\langle #1\rvert}
\newcommand{\braket}[2]{\langle #1\mid #2\rangle}
\newcommand{\abs}[1]{\left\lvert #1\right\rvert}
\newcommand{\norm}[1]{\left\lVert #1\right\rVert}
\def\ketbra#1{\ket{#1}\@ifnextchar\bgroup{\ketbra@two}{\bra{#1}}}
\def\ketbra@two#1{\bra{#1}}
\makeatother
\DeclareMathOperator{\Tr}{Tr}

\usepackage[shortlabels]{enumitem}
\usepackage[linesnumbered,ruled,vlined]{algorithm2e}
\SetKwInput{kwInit}{Init}

\usepackage[toc,page]{appendix}

\usepackage[numbers]{natbib}
\usepackage[hidelinks]{hyperref}

\renewcommand{\geq}{\geqslant}
\renewcommand{\leq}{\leqslant}

\def\calH{\mathcal{H}}
\def\calL{\mathcal{L}}

\newtheorem{theorem}{Theorem}
\newtheorem{lemma}{Lemma}

\newtheorem{definition}{Definition}

\NewDocumentEnvironment{theoremrestated}{o m}{
  \begingroup
  \renewcommand{\thetheorem}{\ref{#2}}
  \addtocounter{theorem}{-1}
  \IfNoValueTF{#1}{\begin{theorem}}{\begin{theorem}[#1]}
}{
  \end{theorem}
  \endgroup
}

\usepackage{authblk}
\usetikzlibrary{quantikz2,fit}
\begin{document}

\title{\Large{\textbf{On the Complexity of Quantum States and Circuits  \\ from the Orthogonal and Symplectic Groups}}}

\author[,1]{Oxana Shaya \thanks{\href{mailto:shayaoxana@gmail.com}{shayaoxana@gmail.com}
}}
\author[2,3]{Zoë Holmes}
\author[1]{Christoph Hirche}
\author[2,3]{Armando Angrisani \thanks{\href{mailto:armando.angrisani@epfl.ch}{armando.angrisani@epfl.ch}}}

\affil[1]{{Institute for Information Processing (tnt/L3S), Leibniz Universit\"at Hannover}, {Germany}}

\affil[2]{{Institute of Physics, Ecole Polytechnique Fédérale de Lausanne}, {Switzerland} }
\affil[3]{{Centre for Quantum Science and Engineering, Ecole Polytechnique F\'{e}d\'{e}rale de Lausanne}, {Switzerland}}
\date{\today}
\maketitle

\begin{abstract}
Understanding the complexity of quantum states and circuits is a central challenge in quantum information science, with broad implications in many-body physics, high-energy physics and quantum learning theory. A common way to model the behaviour of typical states and circuits involves sampling unitary transformations from the Haar measure on the unitary group. In this work, we depart from this standard approach and instead study structured unitaries drawn from other compact connected groups, namely the symplectic and special orthogonal groups. By leveraging the concentration of measure phenomenon, we establish two main results. We show that random quantum states generated using symplectic or orthogonal unitaries typically exhibit an exponentially large strong state complexity, and are nearly orthogonal to one another. Similar behavior is observed for designs over these groups. Additionally, we demonstrate the average-case hardness of learning circuits composed of gates drawn from such classical groups of unitaries.
Taken together, our results demonstrate that structured subgroups can exhibit a complexity comparable to that of the full unitary group.
\end{abstract}

{
  \hypersetup{linkcolor=blue}
  \tableofcontents
}

\section{Introduction}
Random processes lie at the heart of quantum information processing,
as they are essential for probing the properties of many-body quantum systems~\cite{Elben_2022},
the tomography of quantum states~\cite{classical_shadows}, quantum advantage experiments via sampling~\cite{random_q_c, bouland2019quantum}, and quantum cryptography proposals~\cite{ananth2022cryptographypseudorandomquantumstates}. Moreover, they serve as key tools in quantum chaos~\cite{Fisher_2023}, quantum gravity~\cite{Hayden_2007} and quantum machine learning~\cite{cerezo2021variational, larocca2024reviewbarrenplateausvariational}.
In these contexts, the average-case properties of quantum states and processes have traditionally been studied using the Haar measure on the unitary group $U$, which can be informally understood as the uniform distribution over all unitary transformations~\cite{mele2024introduction}.

In this work, we move beyond this extensively studied setting and consider unitary transformations from three compact and connected unitary subgroups: the special unitary group\footnote{In all the applications considered in this paper, sampling unitaries from the special unitary group is equivalent to sampling from the full unitary group. This is because the two groups differ only by a global phase, and we always work with the unitary channel $U(\cdot)U^\dag$, which is insensitive to global phases. Nevertheless, we state our results for all three classical compact and connected subgroups of the unitary group for completeness.}, the unitary symplectic group, and the special orthogonal group. 
These groups are referred to as the \textit{classical} compact and connected groups and naturally arise in various contexts in quantum information theory including the study of (real) randomised benchmarking~\cite{Hashagen2018realrandomized} and symmetry-respecting classical shadows~\cite{west2024realclassicalshadows, west2026classicalshadowsarbitrarygroup}. These constitute more efficient protocols for benchmarking and tomography when the input states or circuits exhibit special symmetries.
Group $k$-\emph{designs} are ensembles of unitaries whose statistical moments match those of the Haar measure up to the $k$-th moment  \cite{mele2024introduction}.
Although these protocols often merely require $2$-designs, higher order designs can be used for randomized benchmarking in order to experimentally characterize the higher order properties of quantum noises \cite{Nakata_2021}.
Even though real-valued unitaries cannot be completely pseudorandom they can still exhibit some pseudorandom properties \cite{brakerski2024realvaluedsomewhatpseudorandomunitaries}.
Results for the special orthogonal group are generally interesting in the context of real
encodings of complex computations~\cite{rudolph20022rebitgateuniversal}.
Recent work in the many-body literature also highlights that ``Haar-like'' output statistics can emerge
in structured settings already at shallow depth and finite system size. In particular, Ref.~\cite{sauliere2025universality}
derives analytical predictions on anticoncentration for random tensor-network ensembles and tests a universal
finite-depth crossover for both unitary and orthogonal brickwork circuits. 
 Moreover, additive-error state designs over the orthogonal and symplectic groups are achievable in $\mathcal{O}(\log\left(\frac{n}{\varepsilon}\right))$ depth in the superblock architecture \cite{grevink2025glueshortdepthdesignsunitary}. This suggests that such ensembles may be accessible with comparatively shallow near-term circuits. 
A hardware-native implementation to approximate the orthogonal group relies on combining the entangling CNOT gate with continuous real-valued single-qubit rotations, such as $R_y(\theta)$. This combination densely generates the required real-subspace rotations~\cite{shi2003toffoli}.
Natively generating the unitary symplectic group on a one-dimensional lattice requires a spatially dependent assignment of Pauli operators to prevent the global circuit from breaking symmetry and leaking into the broader special unitary group \cite{garcia2024architectures}. Specifically, the first pair of adjacent qubits must be driven by locally symplectic generators from the set $\{X_1, Y_1, Y_2, X_1 X_2\}$. To preserve the global symplectic structure, all subsequent adjacent qubit pairs in the chain must be restricted to special orthogonal generators, utilizing the set $\{Y_i, Y_{i+1}, X_i Y_{i+1}, Y_i X_{i+1}\}$.
This provides additional motivation for
understanding typical-state and typical-output properties for orthogonal ensembles, and more broadly for ensembles beyond unitary designs.

Our motivation for considering these subgroups is to understand whether their associated ensembles can replicate average-case properties typically attributed to Haar-random unitaries. As a first step in this direction, we consider the notion of
strong state complexity,
 which measures the complexity of a state by the distance to the computationally trivial state, the maximally mixed state~\cite{Brand_o_2021}.
Most pure quantum states, i.e., typical states generated by random unitaries, have exponentially large strong state complexity and therefore are not physical as they can only be produced after an exponentially long time
\cite{Poulin_2011}.
Our first result shows that this phenomenon extends past the unitary group: states prepared via random unitaries from all the classical compact and connected groups also exhibit this exponential complexity. Moreover, they are geometrically maximally separated.
We refer to Theorem~\ref{thm:high_complexity} for a formal statement of this result. Our results also extend to $k$-designs over the groups. Thus, the exponential high complexity of typical states persists even for these structured subgroup ensembles.
This phenomenon has broad implications:
the complexity of random states plays a key role in modelling quantum chaos~\cite{Balasubramanian_2022}, quantum gravity~\cite{baiguera2025quantumcomplexitygravityquantum} and condensed matter~\cite{Chen_2010}. Furthermore, the fact that most
real states have
exponentially large strong state complexity aligns with the result that imaginarity of states cannot be efficiently tested and is therefore not necessary for the generation of pseudorandom quantum states~\cite{prs}, which are states that are indistinguishable for a quantum polynomial-time adversary from random states sampled uniformly from the entire Hilbert space.

The fundamental role of average-case analysis extends beyond state complexity to the task of sampling from the output distribution of random quantum circuits, a cornerstone of many proposals for demonstrating quantum advantage~\cite{mullane2020sampling}. From the perspective of learning theory, a natural question arises: what is the complexity of learning the output distribution of a quantum circuit? In particular, this involves understanding how many queries are needed, assuming black-box access to the circuit, to approximate its output distribution. This approximation is restricted to accessing only aggregate statistical properties, as formalized in the statistical query (SQ) framework~\cite{kearns1998efficient, nietner2023average}. Motivated by this question, and in analogy to~\cite[Theorem 2]{nietner2023average}, we investigate the average-case complexity of learning output distributions of quantum circuits in the SQ model, focusing on the setting where circuits are drawn from subgroups of the unitary group. We find that the output distribution of circuits drawn from the classical compact groups is average-case hard to learn, which is formally stated and proven in Theorem~\ref{th:average_case_hardness}.

A key structural ingredient in establishing average-case hardness is that the output distributions of such circuits are far from uniform. This feature, in turn, has several practical applications: Aaronson and Chen~\cite{aaronson2016complexity} introduced \emph{heavy output generation} as a task for verifying quantum advantage, which is believed to be infeasible without true quantum sampling. We show that the output distributions induced by ensembles of special orthogonal and symplectic unitaries are indeed far from uniform, making them suitable candidates for performing this task.
Beyond learning and verification, the average-case hardness of learning output distributions of quantum circuits also connects to foundational areas such as quantum cryptography~\cite{fefferman2025hardnesslearningquantumcircuits, wadhwa2024learningquantumprocessesquantum} and the study of black hole dynamics~\cite{yang2023complexitylearningpseudorandomdynamics}, where similar complexity-theoretic barriers emerge.

In the context of quantum machine learning, quantum circuits have also been proposed as a model class for generative modeling, known as
\emph{quantum circuit Born machines} (QCBMs)~\cite{benedetti2019parameterized,coyle2020born}.
A QCBM specifies a parameterized $n$-qubit circuit $U(\theta)$ acting on $\ket{0}^{\otimes n}$, whose
computational-basis measurement induces the model distribution
\[
p_\theta(x)\;=\;\bigl|\bra{x}U(\theta)\ket{0}^{\otimes n}\bigr|^2,\qquad x\in\{0,1\}^n.
\]
The goal of training is to find parameters $\theta$ such that $p_\theta$ approximates
$p_{\mathrm{data}}$ in a chosen statistical distance. Since a QCBM is implemented on quantum hardware, one typically has efficient access to samples from $p_\theta$, rather than to the probabilities themselves. Accordingly, many practical training objectives and gradient estimators are based on empirical averages of bounded test functions evaluated on samples drawn from $p_\theta$ and $p_{\mathrm{data}}$. Therefore, training relies on finitely many measurement shots which fits naturally into the SQ framework.

A crucial modeling point is that the data-generating mechanism—or the trained ansatz family itself—typically induces an \emph{ensemble} of circuits that is effectively random.
A growing literature has linked such randomness assumptions, to obstructions on the training and generalization of QML models - exponential concentration of kernel values when the data-induced unitary ensemble approaches a $2$-design~\cite{thanasilp2024exponential}
and barren plateaus in QCBMs trained with explicit losses~\cite{Rudolph2024}.
A natural attempt to evade these barriers is to inject symmetry into the ansatz—restricting circuits to a structured subgroup such as $SO(D)$ or $Sp(D/2)$, in the spirit of geometric quantum machine learning~\cite{schatzki2024theoretical, ragone2023representationtheorygeometricquantum}. 
 We establish SQ hardness for Haar-random orthogonal and symplectic ensembles.
Our results show that this strategy does not, by itself, recover learnability of QCBMs and respectively trainability of QCBMs with a faithful loss function. 
Our results suggest that excessively expressive ansatze should be avoided and they motivate the adoption of warm-starts, a topic which is currently actively studied~\cite{Puig_2025}.

\paragraph*{Related works.}

The complexity of learning various Born distributions has been studied in previous work. For matchgate circuits, which are isomorphic to the special orthogonal group of dimension $2n$, worst-case hardness of learning has been established in Ref.~\cite{nietner2024freefermiondistributionshard}.
For quantum circuits composed of two-qubit Clifford gates and at super-logarithmic depth, the output distribution is worst-case hard to learn using statistical queries~\cite[Theorem 4]{hinsche2022single}. In contrast, with sample access, the task is efficiently solvable, yet becomes worst-case hard when injecting a single T-gate into at least linear depth Clifford circuits~\cite[Theorem 3]{hinsche2022single}.

Ref.~\cite{nietner2023average} investigates the output distributions of generic quantum circuits composed of randomly drawn 2-local unitaries (with respect to the Haar measure on $U(4)$) arranged in a Brickwork-like architecture.
At depth $d = 1$, the output distributions are product distributions and thus easily learnable. At super-logarithmic depth, learning even a constant fraction of the distribution class requires a super-polynomial number of queries~\cite[Theorem 6]{nietner2023average}. At linear depth, the second moment of the Brickwork ensemble approximates that of the full unitary group~\cite{Haferkamp_2021}, and learning even an exponentially small fraction of the distributions requires exponentially many queries. In the infinite-depth limit, the ensemble converges to the Haar measure over the full unitary group, allowing it to model arbitrary distributions. Consequently, learning a doubly exponentially small fraction of these distributions requires doubly exponentially many queries~\cite[Theorem 2]{nietner2023average}.

Very recently, West established analogous average-case SQ hardness results for Born distributions arising from the
\emph{circular} and \emph{(fermionic) Gaussian} ensembles~\cite{west2026circularGaussianSQ}. This complements the present work, which focuses on ensembles induced
directly by the classical compact \emph{groups} $SO(D)$ and $Sp(D/2)$.

In this work, we utilize Gaussian integration
to determine averages over the classical compact groups.
Recent results appeared during the completion of this work show that the symplectic group induces an exact state design~\cite{west2024randomensemblessymplecticunitary} and that the $D$-dimensional special orthogonal group forms for $t < \sqrt{\varepsilon D}$ an $\varepsilon$-approximate unitary $t$-design~\cite{schatzki2024randomrealvaluedcomplex}.

These findings imply that ensembles drawn from these groups can simulate certain statistical properties of Haar-random unitaries, even at finite depth. Consequently, they can be used to establish the computational hardness of learning the corresponding Born distributions. Nonetheless, our use of Gaussian integration theorems offers an alternative, conceptually simpler proof technique that avoids reliance on Weingarten calculus~\cite{weingarten1978asymptotic, Collins_2006} and may be of independent interest.

\paragraph*{Overview of results and proof ideas.}
We study Haar-induced ensembles over the classical compact groups $G\in\{SU(D),SO(D),Sp(D/2)\}$ with $D=2^n$,
and address two operational questions: (i) how hard is it to distinguish a typical Haar-random state
$\ket{\psi}=U\ket{0^n}$ from the maximally mixed state using measurements implementable by circuits of bounded size, and
(ii) how hard is it to learn the corresponding Born distribution $P_U(x)=|\langle x|U|0^n\rangle|^2$ from noisy
measurement data.
Our first set of results shows that, with overwhelming probability, Haar-random orthogonal and symplectic states have
exponentially large \emph{strong state complexity}: no measurement implementable with at most $r$ elementary gates can
distinguish them from maximally mixed beyond a prescribed advantage (Theorem~\ref{thm:high_complexity}). Moreover, high complexity is compatible
with strong geometric separation: one can pack doubly exponentially many such states that are nearly maximally separated
in trace distance (Theorem~\ref{thm:high_complexity_separation}).
Our second set of results concerns learning. We compute explicit constants $M_G$
governing
the average total-variation distance of $P_U$ to the uniform distribution, and combine this with concentration of measure
to obtain an average-case statistical query (SQ) lower bound: learning even a tiny fraction of these Born distributions to
nontrivial accuracy from $\tau$-accurate queries requires doubly exponentially many queries (Theorem~\ref{th:average_case_hardness}).

Technically, the analysis rests on two recurring ingredients.
On the one hand, L\'evy-type concentration inequalities on $SO(D)$ and $Sp(D/2)$ control fluctuations of Lipschitz functions
of $U$ and of the induced state and output distribution.
On the other hand, we use a Gaussian-integration reformulation of Haar averages that treats the unitary, orthogonal, and
symplectic cases in a unified manner; together, these yield explicit average-case bounds without resorting to Weingarten
calculus. The SQ hardness statements then follow by plugging these concentration estimates into a general lower-bounding technique
(Lemma~\ref{lemma:nie}).

\paragraph*{Organization of the paper.}
The remainder of the paper is organized as follows.
Section~\ref{sec:preliminaries} introduces the setting and notation: the classical compact groups $SU(D)$, $SO(D)$ and $Sp(D/2)$,
the induced state and Born-distribution ensembles, and the concentration tools we use throughout (L\'evy-type bounds and
their design variants). We also recall the two operational notions that underlie our results—strong state complexity and
learning in the statistical query (SQ) model.
Section~\ref{sec:main-results} states and discusses the main theorems: high strong state complexity and geometric separation of typical states,
as well as average-case hardness of learning Born distributions; it also introduces the Gaussian-integration technique used
to evaluate the relevant Haar averages.
Section~\ref{sec:discussion} contains a discussion of implications and open questions, with emphasis on the role of design order and on the
interpretation of SQ hardness for learning-inspired applications.
The proofs of all main results, together with additional technical lemmas and intermediate estimates (including the bounds
on the average distance of $P_U$ to the uniform distribution), are deferred to Appendix~\ref{app:proofs}.

\section{Framework}
\label{sec:preliminaries}
\paragraph*{Notation and ensembles.}
We consider an $n$-qubit system with Hilbert space $\mathcal{H}=(\mathbb{C}^2)^{\otimes n}\simeq \mathbb{C}^{D}$ where $D=2^n$.
We focus on the three compact and connected subgroups of the group of $D$-dimensional unitary matrices, namely the special unitary group $SU(D)$, the unitary symplectic group $Sp(\frac{D}{2})$, and the special orthogonal group $SO(D)$. They are defined as follows:
\begin{align}\label{eq:def_groups}
    SU(D) &= \{g \in GL(D, \mathbb{C})\mid g^\dag g = \mathbb{I}_D \land \det g = 1 \},\\
    Sp\left(\frac{D}{2}\right) &= \{g \in GL\left(\frac{D}{2}, \mathbb{H}\right) \mid g^\dag g = \mathbb{I}_\frac{D}{2} \},\\
     SO(D) &= \{ g\in GL(D, \mathbb{R}) \mid g^\dag g = \mathbb{I}_D \land \det g = 1 \},
\end{align}
where $\dag$ denotes the conjugate transpose over the respective field. We identify $\mathbb{H}^{D/2}\cong\mathbb{C}^{D}$ (with $D$ even)  via $a+j b \mapsto (a,b)$ with $a,b\in\mathbb{C}$, under which the quaternionic unitary group becomes the subgroup $Sp(D/2)\subset U(D)$.
For a compact group $G\subseteq U(D)$ we write $\mu_G$ for its Haar measure.
Each group induces
(i) a pure \emph{state ensemble} $\mathcal{S}_G=\{U\ketbra{0}^{\otimes n}U^\dag\mid U \in G \}$ and
(ii) a \emph{Born distribution ensemble} $\mathcal{P}_G=\{P_U \mid U \in G \}$ where
\[
P_U(x)=\left|\bra{x}U\ket{0}^{\otimes n}\right|^2,\qquad x\in\{0,1\}^n.
\]
We write $\mathcal{U}$ for the uniform distribution on $\{0,1\}^n$.
We denote \(r \simeq s\) to indicate that \(r\) and \(s\) have for large dimensions the same asymptotic scaling.
\paragraph*{Concentration of measure.}
A central tool to prove our results is Lévy's lemma on the classical compact groups, which states that functions on these groups are concentrated around their mean.

\begin{theorem}[Lévy's Lemma]
\label{th:levi}
See \cite[Theorem 5.16 and Theorem 5.5]{Meckes_2019}. Let $G\in\{SU(D),Sp(\frac{D}{2}),SO(D)\}$ and let $f:G\to \mathbb{R}$ be $L$-Lipschitz w.r.t.\ the Hilbert-Schmidt distance, i.e. $ \forall V_1, V_2 \in G :\abs{f(V_1)-f(V_2)}\le L\|V_1-V_2\|_2$.
Then for all $\tau\ge 0$,
\[
\Pr_{U\sim \mu_G}\!\left(\bigl|f(U)-\mathbb{E}_{U\sim \mu_G}[f(U)]\bigr|\ge \tau\right)\le
2\exp\!\left(-\frac{\tau^2}{2L^2 C_G}\right),
\]
where $C_{SO}= \frac{4}{D-2}$, $C_{Sp}=\frac{1}{\frac{D}{2}+1}$ and $C_{SU}=\frac{2}{D}$.
\end{theorem}
The largest upper bound is obtained for $C_G = C_{SO}$ which is therefore the constant we consider when making statements about all unitary subgroups combined.

\smallskip
\noindent \textbf{Group $k$-designs.}
Unitary $k$-designs are sets of unitary matrices that are evenly distributed in the sense that the average of any $k$-th order polynomial over the design equals the average over the entire group. We can phrase the definition in terms of monomials.

\begin{definition}[Definition 2.2 in Ref.~\cite{low2009large}]
A \emph{monomial} in elements of a matrix $U = [U_{i,j}] \in \mathbb{C}^{D \times D}$ is any expression of the form
\begin{align}
    M(U) = \prod_{k=1}^m U_{i_k, j_k}^{a_k} \overline{U}_{i'_k, j'_k}^{b_k},
\end{align}
 where $a_k, b_k \in \mathbb{N}$, and the indices $i_k, j_k, i'_k, j'_k \in \{1, \dots, D\}$.

A balanced monomial of degree $k$  is a monomial containing $k$
conjugated elements and $k$ unconjugated elements.

A polynomial is of degree $k$ if it is a
sum of balanced monomials of degree at most $k$.
\end{definition}

For example, under this definition, $U_{ij}U^{*}_{pq}$ is a balanced monomial of degree $1$ and $U_{ij}^2U^{*}_{pq}U^{*}_{p'q'}$ is a balanced monomial of degree $2$. We can now use this definition to define approximate $k$-designs.

\begin{definition}[Definition 2.6 in Ref.~\cite{low2009large}]
Let $\mu, \nu$ be distributions over $U(D)$.
    Then $\nu$ is an $\varepsilon$-approximate $k$-design with respect to $\mu$ if, for all balanced monomials $M$ of degree $\leq k$,
    $\abs{\mathbb{E}_{U\sim\nu} M(U) - \mathbb{E}_{U\sim\mu} M(U) } \le \frac{\varepsilon}{D ^k}$.
\end{definition}

 We denote the state ensemble induced by sampling $U$ from an $\varepsilon$-approximate unitary $k$-design over $G$ by  $\mathcal{S}^{(k,\varepsilon)}_G$.
For any balanced polynomial function, concentration bounds over the full unitary group immediately carry over to
$k$-designs over the group. The following theorem was originally considered for the Haar measure over the full unitary group. Yet the proof only uses the closeness of moments over designs to the Haar measure moments (\cite[Lemma 3.4]{low2009large}) and therefore only the definition of $k$-designs in general.

\begin{theorem}[Adapted from Theorem 1.2 in Ref.~\cite{low2009large}] \label{th:large_deviation}
Let $G$ be one of the groups $SU(D), Sp(\frac{D}{2}), SO(D)$.
Let $f$ be a polynomial of degree $K$, $f(U) = \sum_i \alpha_i M_i(U)$ where $M_i(U)$ are balanced monomials of degree $K$. Let $\alpha = \sum_i \abs{\alpha_i}$ and $\mu$ be the Haar measure over $G$. Given that
\begin{align}
    \Pr_{U\sim\mu_G}(\abs{f(U)- \mathbb{E}_{U\sim \mu_G}[f(U)]} \geq \delta) \leq C e^{-a \delta^2}
\end{align}
Then for an $\varepsilon$-approximate $k$-design $\nu$ over $G$
\begin{align}
    \Pr_{U\sim\nu}(\abs{f(U)- \mathbb{E}_{U\sim \nu}[f(U)]} \geq \delta) \leq \frac{1}{\delta^{2m}} \left( C\left( \frac{m}{a}\right)^m +\frac{\varepsilon}{D^k} (\alpha + {\abs{\mathbb{E}_{U\sim \nu}[f(U)]}})^{2m} \right)
\end{align}
for integer $m$ with $2mK \leq k$.
\end{theorem}
To simplify notation we suppress ceiling brackets, implicitly rounding to a nearby integer whenever needed.

\paragraph*{Strong state complexity.}
\cite{Brand_o_2021} proposed an operational measure of state complexity that is based on the difficulty of distinguishing the state from the computationally most useless state, the maximally mixed state. It is illustrated in Figure~\ref{fig:strong}.

\begin{figure}[t]
\centering
\begin{tikzpicture}[remember picture]
\node[inner sep=0pt] (circ) {%
\begin{quantikz}[remember picture,
                 row sep={0.55cm,between origins}, column sep=0.32cm,
                 wire types={q,q,q}]
\lstick[3,brackets=none]{}
  & \gate{H}\gategroup[3,steps=5,
       style={rounded corners,fill=blue!10,
              inner xsep=4pt,inner ysep=6pt,alias=Vbox},background,
       label style={label position=above,anchor=south,yshift=-2pt}]{$V$}
  & \ctrl{1} & \qw      & \gate{H} & \ctrl{1}
    & \meter{} \\
  & \gate{H} & \targ{}  & \ctrl{1} & \gate{T} & \targ{}
    & \meter{} \\
  & \gate{H} & \gate{T} & \targ{}  & \gate{H} & \gate{T}
    & \meter{}
\end{quantikz}%
};
\draw[<->] ([yshift=-10pt]Vbox.south west)
        -- node[below=-1pt,font=\scriptsize] {$\mathrm{size}(V)\le r$}
           ([yshift=-10pt]Vbox.south east);
\node[rounded corners=3pt,fill=gray!35,draw,
      minimum width=1.1cm,minimum height=1.4cm,
      align=center,inner sep=2pt,font=\small,anchor=east]
   (src) at (Vbox.west) {$\rho_0$\\or\\$\rho_1$};
\foreach \dy in {-0.55cm,0pt,0.55cm}{%
  \draw ([yshift=\dy]src.east) -- ([yshift=\dy,xshift=8pt]src.east);}
\node[rounded corners=3pt,fill=gray!10,draw,
      minimum width=1.7cm,minimum height=1.4cm,
      align=center,inner sep=2pt,font=\small,anchor=west]
   (gss) at ([xshift=12pt]circ.east|-Vbox.east) {guess\\$\rho_0$ or $\rho_1$};
\foreach \dy in {-0.55cm,0pt,0.55cm}{%
  \draw[double, double distance=1.2pt]
    ([yshift=\dy]gss.west-|circ.east) -- ([yshift=\dy]gss.west);}
\end{tikzpicture}
\caption[Strong state complexity.]{Illustration of the strong state complexity (Definition~\ref{def:state_complexity}).
A referee supplies either $\rho_1=\ketbra{\psi}$ or $\rho_0=\mathbb{I}_D/D$.
The discriminator applies a circuit $V$ of size at most~$r$, drawn from a
universal finite instruction set $\mathsf{G}\subset U(4)$ (here, for
concreteness, $\{H,T,\mathrm{CNOT}\}$), followed by a computational-basis
measurement on each of the $n$ qubits. A binary decision rule applied to
the $n$ classical outcomes yields the guess.}
\label{fig:strong}
\end{figure}

Fix a universal two-qubit gate set $\mathsf{G}\subset U(4)$.
For example, one may take as generators the Hadamard, phase, and CNOT gates together with the $T$ gate.
We assume $\abs{\mathsf{G}}\le n$. Here $\mathsf{G}$ denotes a finite \emph{instruction set} of elementary one- and two-qubit gates to be placed on the $n$ qubits; in particular, it should not be confused with the full $n$-qubit Clifford group.
By the Solovay--Kitaev theorem, switching to a different universal gate set affects compilation cost only by a constant factor~\cite{dawson2005solovaykitaevalgorithm}.
We let $\mathsf{M}_r$ be the set of POVM elements implementable by combining at most $r$ gates from $\mathsf{G}$.
Following the convention of Ref.~\cite{Brand_o_2021}, we state our results for circuits built from one- and two-qubit gates arranged in a 1D architecture. In that setting when estimating the cardinality of $\mathsf{M}_r$,  the number of possible gate placements at each step is $\mathcal{O}(n)$: there are $n$ possible locations for a one-qubit gate and only $n-1$ nearest-neighbour pairs for a two-qubit gate.
However, the same counting arguments extend essentially verbatim to circuits with more general connectivity and gates with higher locality. Indeed, if arbitrary qubits are allowed to interact, then a $k$-local gate can be placed on
$\binom{n}{k}$ positions, so the number of possible gate placements per step increases from $\mathcal{O}(n)$ in the 1D setting with two-local gates to $\mathcal{O}(n^k)$ in the fully nonlocal one with $k$-local gates.
Consequently, factors such as $n^r$ in our bounds are replaced by $n^{kr}$ up to constants, and the corresponding thresholds in $r$ shift accordingly. Thus, relaxing assumption on the connectivity of the circuit and locality of the gates does not qualitatively alter the counting arguments.

\begin{definition}[Strong $\delta$-state complexity,  Definition 1 in Ref.~\cite{Brand_o_2021}]
\label{def:state_complexity}
Fix $\delta\in(0,1)$.
A pure state $\ket{\psi}\in \mathbb{C}^D$ has strong $\delta$-state complexity at most $r$ if
\[
\beta_{\mathrm{qs}}^\sharp (r, \ket{\psi})
\coloneqq
\max_{ M \in \mathsf{M}_r}
\left| \Tr \!\left( M \left(\ketbra{\psi}-\frac{\mathbb{I}_D}{D}\right)\right)\right|
\ge 1-\frac{1}{D}-\delta,
\]
denoted $\mathcal{C}_\delta(\ket{\psi})\le r$.
\end{definition}

This strong definition of state complexity implies other (weaker) definitions, such as the minimal circuit size required to approximate the state. We denote the number of elementary gates required to implement a unitary $V$ by $size(V)$.

\begin{lemma}[Lemma 6 in Ref.~\cite{Brand_o_2021}] \label{lem:state_comp_stronger}
 Suppose that $| \psi \rangle \in \mathbb{C}^D$ obeys $\mathcal{C}_\delta (| \psi \rangle) \geq r+1$ for some $\delta \in (0,1)$ and $r \in \mathbb{N}$.
Then,
\begin{equation}
\min_{\textrm{size}(V) \leq r}
\frac{1}{2} \left\| | \psi \rangle \! \langle \psi| - V |0 \rangle \! \langle 0|^{\otimes n} V^\dag \right\|_1 > \sqrt{\delta}.
\label{eq:simple_state_complexity}
\end{equation}
\end{lemma}

\paragraph*{Learning in the statistical query model.}
 We consider the problem of learning a \emph{concept class} of distributions $\mathcal{D}$. That is, given an instance $P \in \mathcal{D}$, the algorithm returns another distribution $Q$ which is an $\varepsilon$-approximation of $P$ in total variation distance, i.e.
\begin{align}
    d_{TV}(P,Q) \coloneqq \frac{1}{2}\sum_{x\in\{0,1\}^n} \abs{P(x) - Q(x)}\leq \varepsilon.
\end{align}
Introduced in Ref.~\cite{kearns1998efficient}, the SQ model captures algorithms that use coarse statistical properties of the distribution instead of having access to the i.i.d. samples.
Given a distribution $P$ and tolerance $\tau \geq 0$, the SQ oracle $\mathsf{Stat}_\tau$ receives as input a function $\phi : \{0,1\}^n \rightarrow [-1,1]$ and returns a value $v$ such that
\begin{align}
    \left | \mathbb{E}_{x\sim P} \, [\phi(x)] - v \right| \leq \tau.
\end{align}
The framework is illustrated in Figure~\ref{fig:SQ_learning}.
\smallskip
\noindent
Our choice of the SQ model is motivated both by intrinsic features of the quantum measurement process and, more compellingly, by limitations of near-term quantum devices\ \cite{preskill2018quantum}.
In the quantum-circuit setting, the learner typically interacts with a device by repeatedly preparing a circuit
and measuring a the output state (often in the computational basis), producing i.i.d.\ outcomes $x\sim P$.
Any quantity that can be extracted reliably under finite-shot statistics is therefore a \emph{bounded statistic} of the form
$\mathbb{E}_{x\sim P}[\phi(x)]$ for some post-processing function $\phi:\{0,1\}^n\to[-1,1]$.
With $N$ circuit runs, the empirical average $\frac1N\sum_{i=1}^N \phi(x_i)$ estimates $\mathbb{E}[\phi]$ up to additive
error $\mathcal{O}(1/\sqrt{N})$ with high probability (shot noise), so under a realistic polynomial shot budget one is naturally
restricted to tolerances $\tau \gtrsim 1/\mathrm{poly}(n)$.
Moreover, near-term devices often preclude coherent access to multiple entangled copies, leading to single-copy measurements that can be conveniently modeled in the SQ framework; and experimental noise sources further motivate learning models based
on coarse expectation-value information rather than exact sample access or collective measurements.
This operational viewpoint is closely related to quantum statistical query models studied in the quantum-learning
literature; see, e.g., Refs.~\cite{arunachalam2020,nietner2023average,arunachalam2023,Angrisani2025}.
Equivalently, a $\tau$-accurate SQ can be simulated from $N=O(\tau^{-2}\log(1/\delta))$ i.i.d.\ circuit runs
with failure probability at most $\delta$ via empirical averaging.

\begin{figure}[tb!]
    \centering
    \begin{tikzpicture}[auto]
\node           [rectangle,draw]        (node1)     {Algorithm};
\node [below = 0.5cm of node1] (node2) {$Q $ such that  $ d_{TV}(P_U, Q) \leq \varepsilon$};
\node [above=of node1, rectangle, draw, blue] (input)  { $\texttt{Stat}_{\tau}$}    ;
\draw           [->, thick, blue]             ([xshift=0.1cm]  input.south)  --
                node {$v$ such that  $|\mathbb{E}_{x \sim P_U}\phi(x) -v| \leq \tau$ }                ([xshift=0.1cm]  node1.north);
\draw           [->, thick]             ([xshift=-0.2cm] node1.north) --
                node {$\phi: \{0,1\}^n \to [-1,1]$}                 ([xshift=-0.2cm] input.south);
\draw           [->, thick]      ([xshift=0.1cm]  node1.south)  --
                              ([xshift=0.1cm]  node2.north);
\end{tikzpicture}
\caption{Diagrammatic visualisation of distribution learning in the statistical query framework. The algorithm queries the oracle $\texttt{Stat}_{\tau}$ with a bounded function $\phi$ and obtains a $\tau$-accurate expectation value of $\phi$ under the target distribution $P_U$, which in our context refers to the Born distribution of the state $U\ket{0^n}$. After querying $q$ times, the algorithm returns an $\varepsilon$-approximation of the target distribution.}
\label{fig:SQ_learning}
\end{figure}
We do not require the algorithm to succeed always, but with a certain success probability $\beta$.
A lower bound on the average-case query complexity is provided by the following lemma.

\begin{lemma}[Lemma 1 in Ref.~\cite{nietner2023average}]
\label{lemma:nie}
Given a distribution class $\mathcal{D}= \mathcal{P}_G$ of Born distributions and a fixed distribution $\mathcal{U}$, we define the associated maximally distinguishable fraction as
\begin{align}
    \mathfrak{f} \coloneq \max_{\phi\colon \{0,1\}^n \rightarrow [-1,1] } \Pr_{U \sim \mu_G}(| \mathbb{E}_{x \sim P_U}\phi(x) -  \mathbb{E}_{x \sim \mathcal{U}}\phi(x) | \geq \tau ),
\end{align}
and the probability of trivial learning as
\begin{align}
    \mathfrak{u} \coloneq \Pr_{U \sim \mu_G}\left[ d_{TV}(P_U,\mathcal{U}) \leq \varepsilon + \tau\right].
\end{align}
Then $\varepsilon$-learning a $\beta$ fraction of $\mathcal{D}$ with respect to $\mu_G$ requires at least  $q$ many $\tau$-accurate statistical queries with
\begin{align}
    q+1 \geq \frac{\beta - \mathfrak{u} }{\mathfrak{f} }.
\end{align}
\end{lemma}

\section{Main theorems}
\label{sec:main-results}
Now we are ready to derive our main results.

\paragraph*{High strong complexity and geometric separation.}
Our first result is on the strong state complexity for states evolved under unitaries in the classical compact groups.
The results have significance in the context of pseudorandomness and chaos~\cite{Roberts_2017}. They are consistent with the fact that the corresponding random matrix distributions serve as effective models of quantum chaos~\cite{forkel2022classicalcompactgroupsgaussian}.

\smallskip
Before stating the main bound, let us briefly explain how to read it.
We draw a random state $\ket{\psi}$ from the Haar-induced ensemble associated with $G\in\{SU(D),SO(D),Sp(D/2)\}$, and ask
whether $\ket{\psi}$ can be distinguished from the maximally mixed state using measurements implementable with at most $r$
elementary gates from a fixed finite instruction set.
The parameter $\delta\in(0,1)$ specifies the distinguishing advantage required from the measurement family $\mathsf{M}_r$,
while $r$ is the elementary-gate budget used to implement a POVM element.
Theorem~\ref{thm:high_complexity} upper bounds the probability (over the random draw of $\ket{\psi}$) that \emph{there exists} an
$r$-gate measurement that achieves advantage at least $\delta$; thus, a small right-hand side should be read as:
with overwhelming probability, \emph{no} such $r$-gate measurement can distinguish $\ket{\psi}$ from maximally mixed.
For approximate design ensembles, the bound additionally makes explicit how the design order $k$ controls the regime in which
Haar-like complexity persists.

\begin{theorem}[High strong state complexity]
\label{thm:high_complexity}
\medskip
\noindent\textbf{(i)}
Fix $\delta\in(0,1)$. Then
\begin{align}
\Pr_{\ket{\psi}\sim \mathcal{S}_G}\!\left[\mathcal{C}_\delta(\ket{\psi})\le r\right]
\;\le\;
5D (n+1)^r |\mathsf{G}|^r \exp\!\left(-\frac{(D-2)(1-\delta)^2}{32}\right).
\label{eq:thm_high_complexity_prob}
\end{align}
In particular, until $r \simeq \frac{D}{\log n}$ this probability is exponentially small in $D$.

\medskip
\noindent\textbf{(ii)}
Assume
$\delta \in (0,1-\frac{1}{D}-c)$ for some $c>0$, $k>3$ and $\varepsilon \le 2^{-5n-\frac{2k}{3}\log(\frac{5}{4c})}$.

\begin{align}
    \mathrm{Pr}_{\ket{\psi} \sim \mathcal{S}^{(k,\varepsilon)}_G} \left[ \mathcal{C}_\delta (| \psi \rangle) \leq r \right]
\leq
2 (1-D^{-1}-\delta)^{-\frac{2k}{3}} D (n+1)^r\abs{\mathsf{G}}^r \left( 2
 \left(\frac{32 k}{3(D-2)}\right)^{\frac{k}{3}}+ \varepsilon \left( 1+ \frac{1}{D^{\frac{3}{2}}}\right)^{\frac{2k}{3}}
\right)
\label{eq:design_bounds}
\end{align}
In particular, this bound is polynomially suppressed in $D$ until 
$r \simeq \frac{n}{\log n}$.
\end{theorem}

We illustrate the bound in Equation~\eqref{eq:design_bounds} in Figure~\ref{fig:design_plots_prob}. The restriction on the parameters is necessary to ensure that the probability is bounded non-trivially.
 Reference \cite{schuster2025randomunitariesextremelylow} derived that the circuit depth to achieve  $\varepsilon$-approximate unitary $k$-designs is $\mathcal{O}\left(\log(\frac{n}{\varepsilon})\cdot k \, \text{poly}(\log k )\right)$.
Assuming $\mathcal{O}(n)$ gates per layer, this leads to a gate complexity $r =\mathcal{O}\left( n \log(\frac{n}{\varepsilon})\cdot k \text{poly}(\log k )
\right) $. For $\varepsilon$ exponentially small in $n$ and $k$, as assumed in Theorem \ref{thm:high_complexity} (ii), this results in $r =\mathcal{O}\left( n^2 \log(n)\cdot k^2 \text{poly}(\log k )
\right)$ which is still larger than our derived gate complexity.

While our restriction in $\varepsilon$ is significant it contributes a linear in system size factor for depth required to achieve a $k$-design.
We note that our bound in Equation~\eqref{eq:design_bounds} is consistent with the bound in Ref.~\cite[A10]{Brand_o_2021} for the unitary case that leads to the same statement with $r \simeq \frac{k(n-2\log k)}{\log n}$ yet used another proof method.
Moreover, these high complexity states do not occur clustered but are almost orthogonal to each other. This aligns with the fact that there are doubly exponentially many states that are almost orthogonal~\cite[Appendix B]{Roberts_2017}.

\begin{figure}
    \centering
    \includegraphics[width=0.7\linewidth]{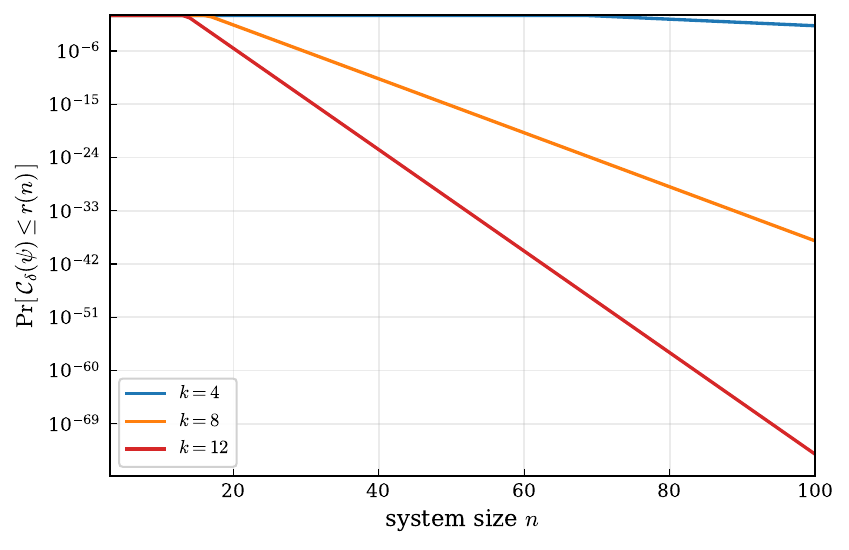}
\caption{Probability upper bound in Theorem~\ref{thm:high_complexity}\,(ii)
versus system size~$n$ (logarithmic $y$-axis, clipped at~$1$).
We set $|\mathsf{G}|=4$, $r(n)=\frac{n}{12\ln n}$,
$\delta=\tfrac{1}{2}$, $c=\tfrac{1}{2}-\tfrac{1}{D}$, 
$\varepsilon=2^{-5n-\frac{2k}{3}\log\!\left(\frac{5}{4c}\right)}$,
and $k\in\{4,8,12\}$.}
    \label{fig:design_plots_prob}
\end{figure}

\smallskip
The next result strengthens the preceding typicality statement by showing that high complexity is compatible with strong
geometric separation.
Rather than considering a single random state, we ask how many states can be selected so that \emph{each} has strong state
complexity at least $r$ and \emph{any pair} is almost maximally far apart in trace distance.
Theorem~\ref{thm:high_complexity_separation} shows that for Haar-induced ensembles one can pack doubly exponentially many such states, while for
approximate design ensembles the achievable packing size is controlled explicitly by the design parameters and the separation
threshold~$\Delta$.

\begin{theorem}[Separation of highly complex states]
\label{thm:high_complexity_separation}

For each of the ensembles below, there exist $N$ pure states
$\ket{\psi_1}, \ldots, \ket{\psi_N}$ in the corresponding ensemble such that
$\mathcal{C}_\delta(\ket{\psi_i}) \ge r$ for all $i$ and their pairwise trace distance satisfies
\begin{align}
\forall i \neq j,\quad \frac{1}{2} \|\ketbra{\psi_i}{\psi_i} - \ketbra{\psi_j}{\psi_j}\|_1 \geq 1 - \Delta.
\label{eq:pairwise_trace_distance}
\end{align}

\medskip
\noindent\textbf{(i) $\mathcal{S}_G$.
}
Let $\delta \in (0,1), \Delta \in (\frac{1}{D}, 1)$ and  $r\lesssim \frac{D}{\log n}$. Then
\begin{align}
N&= \left\lfloor \frac{1}{6}\exp\!\left(\frac{D\Delta^2}{32}\right)\right\rfloor.
\end{align}

\medskip
\noindent\textbf{(ii)
$\mathcal{S}^{(k,\varepsilon)}_G$}.
Let $\delta \in (0,1-\frac{1}{D}-c)$ for $c>0$, $\Delta \geq \frac{1}{D} + 4 \sqrt{\frac{k}{D-2}}$, $k>3$, $\varepsilon \le \min\!\left(2^{-5n-\frac{2k}{3}\log(\frac{5}{4c})},\, 2^{-\frac{nk}{2}-1}\right)$ and $r \lesssim \frac{ n}{\log n}$ (the first $\varepsilon$-bound is inherited from Theorem~\ref{thm:high_complexity}\,(ii), whose hypothesis is used in the proof).
Define
\begin{align}
p_{k,\varepsilon}(\Delta)
&\coloneqq \left(\Delta-\frac{1}{D}\right)^{-k}
\left(
2\left(\frac{16k}{D-2}\right)^{k/2} + 2^k\varepsilon
\right).
\end{align}
Then,
\begin{align}
N&= \left\lfloor \frac{1}{2\,p_{k,\varepsilon}(\Delta)} \right\rfloor.
\end{align}
\end{theorem}

We illustrate the number of geometrically separated pure states drawn from approximate designs in Figure~\ref{fig:design_plots}.
The proof is analogous to the one provided in Ref.~\cite[Appendix A]{Brand_o_2021} and based on Theorem
\ref{th:levi}.
It is presented in the appendix in Sections \ref{sec:proof_high_complexity} and \ref{sec:proof_high_complexity_seperation}.

\begin{figure}[t]
  \centering
  \includegraphics[width=0.49\linewidth]{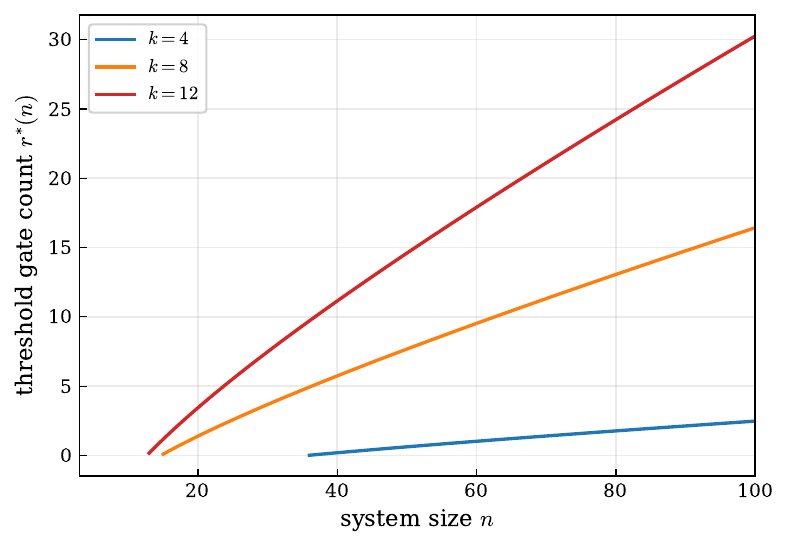}\hfill
  \includegraphics[width=0.49\linewidth]{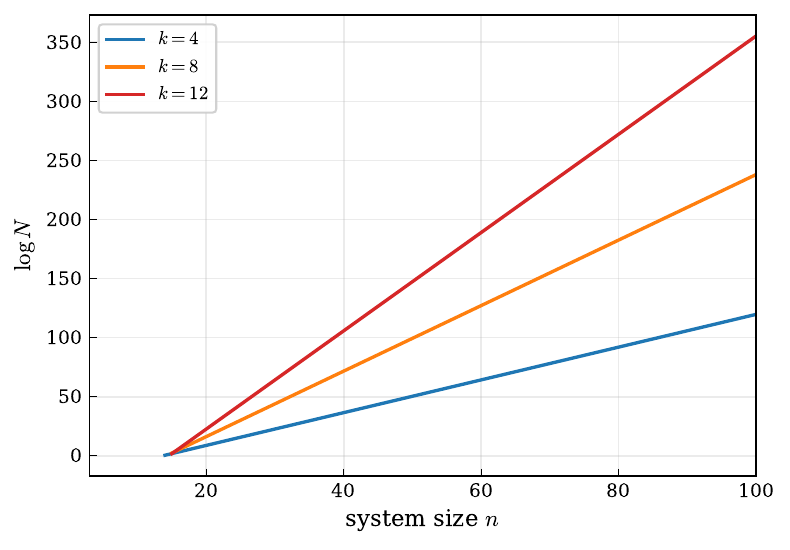}

\caption{Scaling of $k$-design bounds for $k\in\{4,8,12\}$ with
$\delta=\tfrac{1}{2}$ and $c=\tfrac{1}{2}-\tfrac{1}{D}$.
\emph{Left:} threshold gate number~$r^*$ where the
Theorem~\ref{thm:high_complexity}\,(ii) bound on
$\Pr[\mathcal{C}_\delta(\psi)\le r]$ crosses~$1$, with $\varepsilon=2^{-5n-\frac{2k}{3}\log\!\left(\frac{5}{4c}\right)}$;
we assumed $|\mathsf{G}|=4$.
\emph{Right:} $\log N$ from
Theorem~\ref{thm:high_complexity_separation}\,(ii) at fixed
separation $\Delta=0.1$, with $\varepsilon=\min\!\left(2^{-5n-\frac{2k}{3}\log\!\left(\frac{5}{4c}\right)},\, 2^{-nk/2-1}\right)$.}
\label{fig:design_plots}
\end{figure}

\paragraph*{Gaussian integration.}

To leverage Lemma \ref{lemma:nie} for the analysis of the average-case complexity of learning the output distribution of circuits from the classical compact groups, we need to bound the average total variation distance to the uniform distribution.
To analyze moments of functions beyond simple polynomials in the entries of unitaries, we extend the result of~\cite[Theorem 17]{nietner2023average}, originally formulated for the unitary group, to cover all classical compact groups. Specifically, we consider general homogeneous functions of even degree and show that their expectation values over these groups can be reformulated as expectation values with respect to standard independent Gaussian random variables.
We compare the approach for the orthogonal group and the full unitary group in Table~\ref{tab:Gaussian_integration}.
 As $\mathcal{S}_{Sp}$ is an exact unitary state design for all $k$ by~\cite[Theorem 1]{west2024randomensemblessymplecticunitary}, the same formula as in the unitary case holds for the symplectic group.
\begin{table}[htbp]
\centering
\renewcommand{\arraystretch}{1}
\begin{tabular}{|c|c|c|}
\hline
\small \textbf{Group $G$}  & \small \textbf{Haar-random state distribution} & \small \textbf{Gaussian construction} \\
\hline
\small $SO(D)$
& \small Uniform over the real unit sphere $S^{D-1}$
& \small  $g_i \sim \mathcal{N}(0,1)$\\
\hline
\small $U(D)$
& \small Uniform over the complex unit sphere in $\mathbb{C}^D$
& \small $g_i \sim \mathcal{N}(0,1) + i\mathcal{N}(0,1)$\\
\hline
\end{tabular}
\caption{Comparison of Haar-random state generation across classical compact groups.
Each Haar distribution can be realized by normalizing Gaussian vectors over the corresponding field,
and suitable normalization factors given in Theorems~\ref{th:Gaussian_SO} and~\cite[Theorem 17]{nietner2023average}, which convert Gaussian expectations into Haar expectations for homogeneous functions of even degree.}
\label{tab:Gaussian_integration}
\end{table}

\begin{theorem}\label{th:Gaussian_SO}
Let $f : \mathbb{R}^D \to \mathbb{R}$ be homogeneous of even degree $2k$ and let $\mu_{SO}$ denote the Haar-induced measure on real pure states.
Then
\begin{align}
\mathbb{E}_{\ket{\psi}\sim \mu_{SO}}[ f(\ket{\psi}) ]
=
\frac{\mathbb{E}_{g\sim \mathcal{N}(0,\mathbb{I}_D)}[f(g)]}{k!\,2^k\binom{D/2+k-1}{k}}.
\end{align}
\end{theorem}
By a slight abuse of notation, when we write $f(\ket{\psi})$, we mean $f$ is applied to the amplitudes of the state $\ket{\psi} = \sum_{i=1}^D a_i \ket{i}$ in the computational basis. Thus $f(\ket{\psi}) = f(a)$, where $a = (a_1,\dots,a_D) \in \mathbb{R}^D$.
The left-hand side is the Haar average over pure real states, i.e. vectors uniformly distributed on the real unit sphere. On the right-hand side, this same average is rewritten as an expectation over $D$ independent standard Gaussian random variables $g_i \sim \mathcal{N}(0,1)$. Because the direction of a Gaussian vector is uniformly distributed on the sphere and independent of its radius, normalizing such a Gaussian vector reproduces the Haar distribution. The prefactor $\tfrac{1}{k! 2^k \binom{D/2 + k - 1}{k}}$ comes from integrating out the radial part. The proof is in the appendix in section \ref{sec:proof_Gaussian_SO}. With this technical tool we can bound the expected total variation distance
between the Born distribution and the uniform distribution.
This introduces constants $M_G$ and $\Delta_G$ such that
\begin{align}
    M_G-\Delta_G\le \mathbb{E}_{U\sim\mu_G}\, d_{TV}(P_U,\mathcal{U}) \le M_G+\Delta_G.
\end{align}
In particular, $M_{SO}=\sqrt{\tfrac{2}{\pi e}}$ and $M_{SU}= M_{Sp}=\tfrac{1}{e}$, while $\Delta_G=O(D^{-1/2})$. The constants are derived in the appendix in Section \ref{sec:proof_constants}.

\paragraph*{Average-case hardness of learning Born distributions.}

Having established (i) concentration of the Born distribution $P_U$ around an explicit reference distribution for $U$ drawn from $G\in\{SU(D),SO(D),Sp(D/2)\}$ and (ii) the statistical query (SQ) framework for noisy access to $P_U$,
we can now combine these ingredients with Lemma~\ref{lemma:nie} to obtain an average-case hardness result for learning $\mathcal{P}_G$.

\smallskip
We now turn from state distinguishability to the task of learning the corresponding Born output distributions.
The SQ learner interacts with the target distribution $P_U$ only through $\tau$-accurate estimates of expectations
$\mathbb{E}_{x\sim P_U}[\phi(x)]$ of bounded test functions $\phi:\{0,1\}^n\to[-1,1]$, which models finite-shot access and
measurement noise.
Our lower bound combines (i) the fact that a typical $P_U$ is far from uniform in total variation distance on average, and
(ii) concentration-of-measure bounds controlling the fraction of circuits for which this distance is atypically small,
together with the general SQ template of Lemma~\ref{lemma:nie}.
Theorem~\ref{th:average_case_hardness} gives the resulting average-case query-complexity lower bound, expressed in terms of
the tolerance $\tau$, the learning accuracy $\varepsilon$, and the parameter $\xi_G:=M_G-\Delta_G-(\varepsilon+\tau)$.
\begin{theorem}\label{th:average_case_hardness}
Fix $\varepsilon,\tau\ge 0$ and set $\xi_G \coloneqq M_G-\Delta_G-(\varepsilon+\tau)\ge 0$, where $(M_G,\Delta_G)$ are defined as above.
Any algorithm that $\varepsilon$-learns a $\beta$-fraction of Born distributions $\mathcal{P}_G$ with $\tau$-accurate statistical queries must use $q$ queries satisfying
\begin{align}
q+1 \;\ge\;
\begin{cases}
\displaystyle
\frac{\beta - 2 \exp\!\left(- \frac{(2^{n}-2)\,\xi_{SO}^2}{8}\right)}
     {2 \exp\!\left(-\frac{ (2^n -2)\,\tau^2}{32}\right)}
&\text{for } G=SO(D),
\\[1.25em]
\displaystyle
\frac{\beta - 2 \exp\!\left(- \frac{(2^{n-1}+1)\,\xi_{SU}^2}{2}\right)}
     {2 \exp\!\left(-\frac{ (2^{n-1} +1)\,\tau^2}{8}\right)}
&\text{for } G=Sp(D/2),
\\[1.25em]
\displaystyle
\frac{\beta - 2 \exp\!\left(- \frac{2^n \,\xi_{SU}^2 }{4}\right)}
     {2 \exp\!\left(- \frac{2^n \,\tau^2 }{16}\right)}
&\text{for } G=SU(D).
\end{cases}
\end{align}
\end{theorem}
The complete proof can be found in the appendix in Section \ref{sec:proof_hardness}.
The denominator of the query complexity lower bound is still doubly exponential in the number of qubits for $\tau^2 \geq 2^{-\frac{n}{2}}$.
The numerator is still exponentially close to one for $\xi_G^2 \geq 2^{- \frac{n}{2}}$,  $\beta = 2^{-2^{\Omega(n)}}$.
This restricts the accuracy to $\varepsilon \leq M_G - \Delta_G - 2\tau$.
That is for any $\tau^2 \geq 2^{-\frac{n}{2}}, \varepsilon \leq M_G- \Delta_G - 2\tau$ and $\beta = 2^{-2^{\Omega(n)}}$
learning the output distribution of random quantum circuits over the classical connected compact groups requires at least doubly exponentially many queries $q = 2^{2^{\Omega(n)}}$. We illustrate the scaling in Figure~\ref{fig:SQ}.
\begin{figure}
    \centering
    \includegraphics[width=0.5\linewidth]{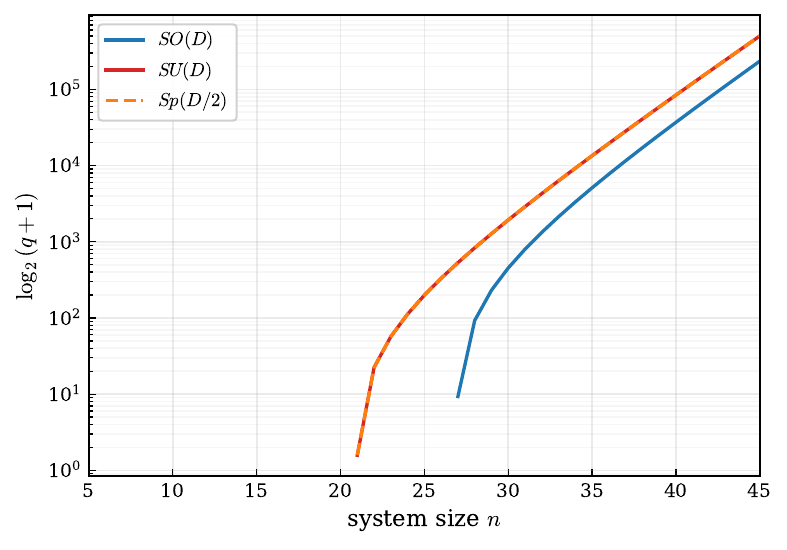}
    \caption{Lower bounds on the required number of statistical queries, plotted as $\log(q+1)$ versus $n$, for $G\in\{SO(D),SU(D),Sp(D/2)\}$,  $\tau=2^{-n/4}$, $\varepsilon=M_G-\Delta_G-2\tau$ (hence $\xi_G=\tau$), and $\beta(n)=2^{-2^{n/3}}$.}
    \label{fig:SQ}
\end{figure}
\section{Discussion and future work}\label{sec:discussion}
In this work, we discussed the complexity of states and unitaries
sampled from structured subgroups of the unitary group.
By applying Lévy's Lemma, we showed that almost all states evolved under unitaries drawn from the classical compact groups exhibit high circuit complexity. For unitary designs of order $k > 3$ over these groups, we derived lower bounds on the circuit depth required to generate them. This opens several directions for future research, such as
studying the growth of complexity under restricted, non-universal gate sets, as in Ref.~\cite{Haferkamp_2022}. Furthermore, again leveraging Lévy's Lemma, one could show that state ensembles generated by the classical compact groups exhibit thermalization behavior, following the framework of Ref.~\cite{Hayden_2006, Popescu_2006}. In addition, it would be interesting to explore the hardness of random circuits that are far from approximate $k$-designs.
Another interesting direction is to analyse the interplay of noise and randomness from restricted groups. We expect unital noise to give even stronger concentration phenomena, but we leave the analysis to future works. 
We have contributed to clarifying the similarities among the classical compact groups. While Ref.~\cite{grevink2025glueshortdepthdesignsunitary, west2025nogotheoremssublineardepthgroup} emphasized their differences in the construction of designs, our results highlight structural features they share. A natural direction for future work is to provide a unified explanation of both their commonalities and divergences, thereby deepening the understanding of why these groups behave alike in some aspects yet differ in others.
\section*{Acknowledgements}
  OS thanks the EPFL Center for Quantum Science and Engineering for supporting this work in its initial state through the INSPIRE Quantum Master Award.
  ZH and AA acknowledge support from the Sandoz Family Foundation-Monique de Meuron program for Academic Promotion.
  The authors thank Sreejith Sreekumar, Mario Berta, Alexander Nietner, Marios Ioannou and Marcel Hinsche for fruitful discussions.
\bibliography{arxiv1}
\bibliographystyle{unsrtnat}
\clearpage
\newpage
\appendix
\section{Proofs}
\label{app:proofs}
\subsection{Preliminaries}
We start by providing some essential mathematical preliminaries.
We denote the $D$-dimensional identity matrix by $\mathbb{I}_D$.

\smallskip
\noindent\textbf{Fields and groups.}
 A \emph{field} $\mathbb{F}$ is a set equipped with operations of addition, subtraction, multiplication, and division, which satisfy the axioms corresponding to those operations on the rational or real numbers. A \emph{group} $G$ is a set equipped with a binary operation $(g, h) \mapsto gh$ that is associative, admits an identity element, and in which every element possesses an inverse with respect to the group operation.
A group is called \emph{topological} if it is also a topological space and the group operations, i.e. the binary operation $(g, h) \mapsto gh$ and inversion $g \mapsto g^{-1}$, are continuous with respect to the topology. A topological matrix group $G$ is compact if it is closed and bounded and it is \emph{connected} if there is a continuous path within $G$ between any two elements. A topological space is called \emph{separable} if it contains a countable dense subset.
The \emph{general linear group} over the field $\mathbb{F}$, denoted $GL(D, \mathbb{F})$, is defined as the group of all invertible $D \times D$ matrices with entries in $\mathbb{F}$, under matrix multiplication. That is, it consists of all matrices with nonzero determinant and entries in $\mathbb{F}$.
We denote by $\mathbb{R}$, $\mathbb{C}$, and $\mathbb{H}$ the field of real numbers, the field of complex numbers, and the field of quaternions, respectively.

\smallskip
\noindent\textbf{Bra-ket notation.} Let $\{\ket{0},\ket{1}\}$ be the canonical basis of $\mathbb{C}^2$, and $\mathcal{H}_n =(\mathbb{C}^2)^{\otimes n}$ be the Hilbert space of $n$ qubits.
We use the bra-ket notation, where we denote a vector $v \in \mathcal{H}_n$ using the ket notation $\ket{v}$ and its adjoint using the bra notation $\bra{v}$. For $u,v\in \mathcal{H}_n$, we will denote by $\braket{u}{v}$ the standard Hermitian inner product $u^\dag v$. A pure state is a normalized vector $\ket{v}$, i.e. $|\braket{v}|=1$.

\smallskip
\noindent \textbf{Schatten Norms.}
For a vector $\ket{\psi}\in\mathcal{H}_n$ with entries $\psi_i$ and $p\in[1,\infty]$, the $p$-norm is defined as
\begin{align}
\|\ket{\psi}\|_p \;=\;\biggl(\sum_i|\psi_i|^p\biggr)^{1/p}.
\end{align}
Let $\calL(\calH_n)$ denote the set of linear operators $A:\calH_n \rightarrow \calH_n$.
For an operator $A\in\mathcal{L}(\mathcal{H}_n)$ and $p\in[1,\infty]$, the Schatten $p$-norm $\|A\|_p$ is the (vector) $p$-norm of the vector of singular values of $A$.
For $p=1$, it is often referred to as \emph{trace norm}, while for $p=2$ it is often referred to as \emph{Hilbert-Schmidt norm}.
For $p=\infty$, the Schatten $p$-norm is called \emph{operator norm} and can also be written
\begin{align}
\|A\|_\infty \;=\;\sup_{\|\ket{\psi}\|_2=1}\|A\ket{\psi}\|_2.
\end{align}
For convenience, we usually drop the subscript for the operator norm, i.e.
\begin{align}
\|A\|\;\coloneqq\;\|A\|_\infty.
\end{align}
Vector $p$-norms and Schatten $p$-norms are monotone with respect to the index $p$:
\begin{align}
    \norm{A}_p \geq \norm{A}_q \qquad \text{if $1\leq p\leq q \leq \infty$}. \label{eq:monotone}
\end{align}
Moreover, a reverse inequality with an additional dimensional factor also holds:
\begin{align}
    \norm{A}_p \leq D^{\frac{1}{p}-\frac{1}{q}}\norm{A}_q \qquad \text{if $1\leq p\leq q \leq \infty$} \label{eq:reverse},
\end{align}
where we set $D \coloneqq \mathrm{dim}(\mathcal{H}_n)$.
In particular, the operator norm and the Hilbert-Schmidt norm satisfy the following:
\begin{align}
    \norm{A} \leq \norm{A}_2 \leq\sqrt{D}\norm{A}.
\end{align}
Distances induced by Schatten norms play a central role in quantum information theory.
In particular, the distance induced by the Schatten $1$-norm, commonly called the \emph{trace distance}, admits the following form for pure states:
\begin{align}
    \tfrac{1}{2}\bigl\|\ketbra{\psi} - \ketbra{\phi}\bigr\|_1
    = \sqrt{1 - \abs{\braket{\psi}{\phi}}^2}. \label{eq:trace-dist}
\end{align}
Furthermore, for an operator $A$ acting on a $D$-dimensional Hilbert space, we have
\begin{align}\label{eq:alpha}
     \sum_{i,j} \abs{ A_{ij}}
     \leq D \sqrt{\sum_{i,j} \abs{ A_{ij}}^2} \eqqcolon D \norm{A}_2 \leq D \sqrt{D} \norm{A},
\end{align}
where we used the inequality in Equation~\eqref{eq:reverse} twice, first for the vector of the entries of $A$ with dimension $D^2$ and second for $A$ itself, and we noted that the $2$-norm of that vector coincides with the Schatten 2-norm of $A$.

\smallskip
\noindent\textbf{Haar measure.}
Compactness is crucial to define a measure. For the classical compact groups, there exists a unique normalised invariant Radon measure, that is, for all $A \subset G, g \in G$:
\begin{align}
    \mu(A) = \mu(\{gh| h \in A\}) = \mu(\{hg| h \in A\}).
\end{align}
This invariant measure is also called the Haar measure as introduced in Ref.~\cite{haar1933massbegriff}.

\smallskip
\noindent
\textbf{Computing moments of the Haar measure.}
In order to compute the moments of the Haar measure, we consider the notion of \emph{group commutant}.
Denote the space of linear operators acting on $n$ qubits as $\mathcal{B}$. The commutant of the $k$-fold tensor representation of a group $G$
refers to the set of all linear operators that commute with every element of the group acting on the $k$-fold tensor product space, i.e.
\begin{align}
    \text{comm } (G^{\otimes k}) \coloneq  \{ A \in \mathcal{B}^{\otimes k} \, | \, \forall U \in G : [A,U^{\otimes k}] = 0 \}
\end{align}
The $k$-th moments can be calculated via the projector into the commutant of the $k$-th fold tensor representation of the group $G$~\cite{diaz2023showcasing}.
\begin{align}\label{eq:moments}
     \mathcal{E}^{(k)}_G&\colon \mathcal{B}^{\otimes k} \to \text{comm } (G^{\otimes k})\\
     &A \mapsto \mathbb{E}U^{\otimes k} A U^{\dag\,\otimes k} = \int_G dU U^{\otimes k} A U^{\dag\,\otimes k} =
     \sum_{\eta = 1}^{\dim(\text{comm}(G^{\otimes k}))} \Tr(B_{\eta}^{(k)}A) B_{\eta}^{(k)}
\end{align}
where $\{B_{\eta}^{(k)}\}_{\eta =1 }^{\dim \text{comm}(G^{\otimes k})}$ is a Hermitian, orthonormal basis of the commutant of the $k$-th fold tensor representation of $G$.
For the classical compact groups, the structure of the commutant of the $k$-fold tensor representation is characterized by Schur-Weyl duality~\cite{goodman2000representations}. Specifically, for the unitary group $U(D)$, the commutant of its $k$-fold tensor representation is spanned by a representation of the symmetric group $S_k$, whose elements correspond to permutations~\cite{harrow2023approximate}. In contrast, for the orthogonal and symplectic groups, the commutant is spanned by representations of the Brauer algebra  $\mathfrak{B}_k(D)$ and $\mathfrak{B}_k(-D)$, respectively~\cite{brauer1937algebras}. These algebras consist of all pairings of a set of size $2k$ and in particular contain the symmetric group as a subalgebra.
Despite these structural differences, the commutant of the first tensor power is trivial for all three groups, containing only the identity, implying that $SO(D)$ and $Sp(\frac{D}{2})$ are unitary $1$-designs.
\subsection{Proof of Theorem~\ref{thm:high_complexity}}\label{sec:proof_high_complexity}
First, we make an observation that lets us transfer Lipschitz constants.
We note that any function on the state space, $f : \mathcal{H}_n \rightarrow \mathbb{R}$, can naturally be viewed as a function on the unitary group, $f : U(2^n) \rightarrow \mathbb{R}$, by defining $f(U) := f(U\ket{0}^{\otimes n})$.
The following lemma will be useful for bounding the Lipschitz constant of such functions when considered over the unitary group.
\begin{lemma}\label{lemma:lipschitz}
If a function $f$ is $L$-Lipschitz on the state space, then it is also $L$-Lipschitz when viewed as a function on unitaries.
\end{lemma}
\begin{proof}
The desired result follows from the chain of inequalities
\begin{align}
    | f(U_1 \ket{0}^{\otimes n}) - f(U_2 \ket{0}^{\otimes n}) |
    &\leq L \norm{(U_1 - U_2) \ket{0}^{\otimes n}}_2 \\
    &\leq
    L \norm{U_1 - U_2} \\
    &\leq L \norm{U_1 - U_2}_{2}.
\end{align}
Here, the first inequality follows from the Lipschitz continuity of $f$ with constant $L$; the second uses the definition of the operator norm, i.e. $\norm{A} = \sup_{\norm{v}_2=1} \norm{A v}_2$, and the final inequality follows from the monotonicity of Schatten norm as stated in Equation~\eqref{eq:monotone}.
This completes the proof.
\end{proof}
Now we are ready to prove Theorem \ref{thm:high_complexity}.
\begin{theoremrestated}[High strong state complexity]{thm:high_complexity}
\medskip
\noindent\textbf{(i)}
Fix $\delta\in(0,1)$. Then
\begin{align}
\Pr_{\ket{\psi}\sim \mathcal{S}_G}\!\left[\mathcal{C}_\delta(\ket{\psi})\le r\right]
\;\le\;
5D (n+1)^r |\mathsf{G}|^r \exp\!\left(-\frac{(D-2)(1-\delta)^2}{32}\right).
\end{align}
In particular, until $r \simeq \frac{D}{\log n}$ this probability is exponentially small in $D$.
\medskip
\noindent\textbf{(ii)}
Assume
$\delta \in (0,1-\frac{1}{D}-c)$ for some $c>0$ $k>3$ and $\varepsilon \le 2^{-5n-\frac{2k}{3}\log(\frac{5}{4c})}$.

\begin{align}
    \mathrm{Pr}_{\ket{\psi} \sim \mathcal{S}^{(k,\varepsilon)}_G} \left[ \mathcal{C}_\delta (| \psi \rangle) \leq r \right]
\leq
2 (1-D^{-1}-\delta)^{-\frac{2k}{3}} D (n+1)^r\abs{\mathsf{G}}^r \left( 2
 \left(\frac{32 k}{3(D-2)}\right)^{\frac{k}{3}}+ \varepsilon \left( 1+ \frac{1}{D^{\frac{3}{2}}}\right)^{\frac{2k}{3}}
\right)
\end{align}
In particular, this bound is polynomially suppressed in $D$ until 
$r \simeq \frac{n}{\log n}$.
\end{theoremrestated}
\begin{proof}
\medskip
Following the steps in Ref.~\cite[Appendix A 1 a]{Brand_o_2021}, in Equation~\eqref{eq:union_bound} we apply a union bound. Next, we apply Lévy's Lemma (Theorem~\ref{th:levi}) in Equation~\eqref{eq:levi_lemma}
to give, after further simplifications, an exponentially (in system size) suppressed upper bound on the probability that a state has low strong state complexity. Note that, as in Ref.~\cite[Proposition 29]{Brand_o_2021}, the function $f_M \colon \ket{\psi}  \mapsto \bra{\psi} M \ket{\psi}$ is Lipschitz with Lipschitz constant $2 \norm{M}_\infty\leq 2$ and by Lemma \ref{lemma:lipschitz} also $2$-Lipschitz as a function on the unitary subgroups. For $G \in \{SO(D), Sp(\frac{D}{2}), SU(D)\}$
we bound the probability of states $\ket{\psi} \sim \mathcal{S}_{G} $ having strong state complexity at most $r$:
\begin{align}
\mathrm{Pr}_{\ket{\psi} \sim \mathcal{S}_{G}} \left[ \mathcal{C}_\delta (| \psi \rangle) \leq r \right]
&= \mathrm{Pr}_{\ket{\psi} \sim \mathcal{S}_{G}}  \left[ \max_{M \in \mathsf{M}_r} \left| \Tr \left( M \left(| \psi \rangle \! \langle \psi| - \frac{\mathbb{I}_D}{D}\right) \right) \right| \geq 1- D^{-1} -\delta \right] \label{eq:union_bound} \\
&\leq \left| \mathsf{M}_r \right| \max_{M \in \mathsf{M}_r} \mathrm{Pr}_{\ket{\psi} \sim \mathcal{S}_{G}} \left[ \left| \Tr \left( M \left(|\psi \rangle \! \langle \psi| - \frac{\mathbb{I}_D}{D} \right) \right) \right| \geq 1- D^{-1} - \delta \right]
\label{eq:levi_lemma}
\\
&\leq 4D(n+1)^r \abs{\mathsf{G}}^r
    \exp \left( - \frac{(D-2) (1- D^{-1} - \delta)^2}{32}\right)
\\&\leq 5D(n+1)^r \abs{\mathsf{G}}^r  \exp \left( - \frac{(D-2) (1- \delta)^2}{32}\right)\label{eq:high_complexity}
\end{align}
 Until  $r \simeq  \frac{D}{\log(n)}$, the probability of a state evolved under unitaries drawn randomly from the classical compact groups not having high state complexity is exponentially suppressed in $D$.
For state ensembles that arise from $\varepsilon$-approximate $k$-designs $\mathcal{S}^{(k,\varepsilon)}_G$, we can derive an analogous bound via Theorem \ref{th:large_deviation}, but this time we find that the upper bound is polynomially (rather than exponentially) suppressed in the system size  under certain conditions. We consider
\begin{align}
    f(U) = \Tr(M \ketbra{\psi}) = \Tr(M U \ketbra{0}^{\otimes n} U^\dag) = \sum_{j, k} M_{jk} U_{k0} U^\dag_{0j},
\end{align}  which is a polynomial of degree one.
POVM measurements are normalized and the operator norm of $M$ is in turn bounded by one, so we get $\alpha(f) = \sum_{ij} \abs{ M_{ij}} \leq D\sqrt{D}$ by using Equation~\eqref{eq:alpha}.
Furthermore,
\begin{align}
    \abs{\mathbb{E}_U f(U) } = \abs{\Tr(M \int d\mu(U) U \ketbra{0}^{\otimes n} U^\dag )} =\frac{1}{D} \Tr(M) \leq 1,
\end{align}
where again we use that $M \leq I$ and hence $\Tr(M) \leq D$.
Again, by Theorem~\ref{th:large_deviation} we have for an integer $m$ such that  $2m \leq k$,
\begin{align}\label{eq:high_complexity_designs}
\mathrm{Pr}_{\ket{\psi} \sim \mathcal{S}^{(k,\varepsilon)}_G}  \left[ \mathcal{C}_\delta (| \psi \rangle) \leq r \right]
\leq
\frac{2D(n+1)^r\abs{\mathsf{G}}^r}{ (1- D^{-1} - \delta)^{2m}}
\left( 2 \left( \frac{32 m}{D-2}\right)^m + \frac{\varepsilon}{D^k} (D^{\frac{3}{2}} +1)^{2m} \right)
\end{align}
Setting $m= \frac{k}{3}$, preventing the
$\varepsilon$-term from blowing up, leads to the following bounds under the reasonable assumptions 
$\delta \in (0,1-\frac{1}{D}-c)$ for $c>0$.
\begin{align}
    \mathrm{Pr}_{\ket{\psi} \sim \mathcal{S}^{(k,\varepsilon)}_G} \left[ \mathcal{C}_\delta (| \psi \rangle) \leq r \right]
\leq
2 (1-D^{-1}-\delta)^{-\frac{2k}{3}} D (n+1)^r\abs{\mathsf{G}}^r \left( 2
 \left(\frac{32 k}{3(D-2)}\right)^{\frac{k}{3}}+ \varepsilon \left( 1+ \frac{1}{D^{\frac{3}{2}}}\right)^{\frac{2k}{3}}
\right)
\end{align}
The probability that a state sampled from an $\varepsilon$-approximate $k$-design for $D>2$,$k>3$ and $\varepsilon \le 2^{-5n-\frac{2k}{3}\log(\frac{5}{4c})}$  having a state complexity less than $r$ stays polynomially suppressed in $D$ until $r\simeq \frac{n }{\log(n)} $.
.
\end{proof}
\subsection{Proof of Theorem \ref{thm:high_complexity_separation} } \label{sec:proof_high_complexity_seperation}
\begin{theoremrestated}[Separation of highly complex states]{thm:high_complexity_separation}
For each of the ensembles below, there exist $N$ pure states
$\ket{\psi_1}, \ldots, \ket{\psi_N}$ in the corresponding ensemble such that
$\mathcal{C}_\delta(\ket{\psi_i}) \ge r$ for all $i$ and their pairwise trace distance satisfies
\begin{align}
\forall i \neq j,\quad \frac{1}{2} \|\ketbra{\psi_i}{\psi_i} - \ketbra{\psi_j}{\psi_j}\|_1 \geq 1 - \Delta.
\end{align}
\medskip
\noindent\textbf{(i) $\mathcal{S}_G$.
}
Let $\delta \in (0,1), \Delta \in (\frac{1}{D}, 1)$ and  $r\lesssim \frac{D}{\log n}$. Then
\begin{align*}
N&= \frac{1}{6}\exp\!\left(\frac{D\Delta^2}{32}\right).
\end{align*}
\medskip
\noindent\textbf{(ii)
$\mathcal{S}^{(k,\varepsilon)}_G$}.
Let $\delta \in (0,1-\frac{1}{D}-c)$ for $c>0$, $\Delta \geq \frac{1}{D} + 4 \sqrt{\frac{k}{D-2}}, k>3, $ $ \varepsilon \le 2^{-5n-\frac{2k}{3}\log(\frac{5}{4c})}$ and $r \lesssim \frac{ n}{\log n}$.
Define
\begin{align*}
p_{k,\varepsilon}(\Delta)
&\coloneqq \left(\Delta-\frac{1}{D}\right)^{-k}
\left(
2\left(\frac{16k}{D-2}\right)^{k/2} + 2^k\varepsilon
\right).
\end{align*}
Then
\begin{align*}
N&= \frac{1}{2\,p_{k,\varepsilon}(\Delta)} .
\end{align*}
\end{theoremrestated}
\begin{proof}
The proof is analogous to the one provided in Ref.~\cite[Appendix A 1.b]{Brand_o_2021}. We can construct an exponentially large set with high complexity states that have almost maximal trace distance via a probabilistic method.
In a first step, we show that on average, a state sampled from these ensembles is far away from any fixed pure state $| \phi \rangle \! \langle \phi|$.
Using that for pure states the trace distance is the square root of one minus the fidelity, see Eq.~\eqref{eq:trace-dist}, we have
\begin{align}
    \mathrm{Pr}_{\ket{\psi} \sim \mathcal{S}_{G}} \left[ \tfrac{1}{2} \left\| \ketbra{\psi} -\ketbra{\phi} \right\|_1 \leq 1-\Delta \right] =\mathrm{Pr}_{\ket{\psi} \sim \mathcal{S}_{G}} \left[ | \langle \psi | \phi \rangle|^2 \geq (2 - \Delta)\Delta \right]
    \leq \mathrm{Pr}_{\ket{\psi} \sim \mathcal{S}_{G}} \left[ | \langle \psi | \phi \rangle|^2 \geq \Delta \right],
\end{align}
where in the last inequality we used that $(2-\Delta)\Delta\geq \Delta $ for $\Delta \in [0,1]$.
Next, we apply Lévy's Lemma to $f_{\ketbra{\phi}} \colon U  \mapsto  \Tr(U \ketbra{0}U^\dag \ketbra{\phi}) $ which is $2 \norm{\ketbra{\phi}}_\infty = 2 $ Lipschitz (\cite[Proposition 29]{Brand_o_2021}). This  leads to $2\exp\left(- \frac{\tau^2}{8 C_G} \right) \geq \mathrm{Pr}[ \abs{\abs{\braket{\psi}{\phi}}^2- \frac{1}{D}} \geq \tau ] \geq \mathrm{Pr}[  \abs{\braket{\psi}{\phi}}^2 \geq \tau + \frac{1}{D} ]$ for $\tau\geq 0 $ .
Now we can choose $\tau = \Delta-\frac{1}{D}$ for any $\Delta \in (\frac{1}{D},1)$, such that
\begin{align}
\mathrm{Pr}_{\ket{\psi} \sim \mathcal{S}_{G}} \left[ | \langle \psi | \phi \rangle|^2 \geq \Delta \right]
&\leq
2
    \exp \left(-\frac{(\Delta - \frac{1}{D})^2(D-2)}{32}\right) \leq 3  \exp \left(-\frac{D\Delta^2}{32}\right)
 \label{eq:far_away_from_fixed}
\end{align}
where basic algebraic manipulations lead to the last inequality.
Using Theorem~\ref{th:large_deviation}, with  $\alpha(f_{\ketbra{\phi}}) \leq D$ and $\mathbb{E}_U f_{\ketbra{\phi}}(U) =  \frac{1}{D}$, we obtain for $\varepsilon$-approximate $k$-designs
\begin{align}
    \mathrm{Pr}_{\ket{\psi} \sim \mathcal{S}^{(k,\varepsilon)}_G} \left[ \tfrac{1}{2} \left\| \ketbra{\psi} -\ketbra{\phi} \right\|_1 \leq 1-\Delta \right] \leq
    \left(\Delta - \frac{1}{D} \right)^{-2m} \left( 2 \left(\frac{32m}{D-2}\right)^{m} +\frac{\varepsilon}{D^k} \left(D +\frac{1}{D}\right)^{2m} \right)
     \label{eq:far_away_from_fixed_designs}
\end{align}
Setting $m = \frac{k}{2}$ to its maximum value leads to the following bounds,
\begin{align}\label{eq:design_bounds_ortho}
      \mathrm{Pr}_{\ket{\psi} \sim \mathcal{S}^{(k,\varepsilon)}_G} \left[ \tfrac{1}{2} \left\| \ketbra{\psi} -\ketbra{\phi} \right\|_1 \leq 1-\Delta \right]
     \leq    \left(\Delta - \frac{1}{D} \right)^{-k}\left(
      2\left(\frac{16k}{D-2}\right)^{\frac{k}{2} } +  2^k\varepsilon \right)
\end{align}
For this probability to be polynomially suppressed in $D$, we require an exponentially small $\varepsilon$, namely $\varepsilon \leq \frac{1}{2}D^{-\frac{k}{2}}$, and additionally $\varepsilon$ must satisfy the hypothesis of Theorem~\ref{thm:high_complexity}(ii) since Equation~\eqref{eq:high_complexity} is invoked below. Furthermore, $\Delta \geq \frac{1}{D} + 4 \sqrt{\frac{k}{D-2}}$ and $k \geq 2$.
For the construction we need that both probabilities described by Equation~\eqref{eq:design_bounds} and Equation~\eqref{eq:design_bounds_ortho} have to be small.
By Equation~\eqref{eq:high_complexity} there exists an $r\lesssim \frac{D}{\log n}$ resp.  $r\lesssim \frac{n}{\log n}$  such that the probability for a state to have low complexity $\mathrm{Pr} \left[ \mathcal{C}_\delta (| \psi \rangle) \leq r \right]$ is strictly smaller than $\frac{1}{2}$.
Equivalently, there is an $r$ such that the probability to sample a high-complexity state is $\mathrm{Pr} \left[ \mathcal{C}_\delta (| \psi \rangle) \geq r \right] \ge \frac{1}{2}$, hence there exists such a state $\ket{\psi_1}$ satisfying $\mathcal{C}_\delta (| \psi_1 \rangle) \geq r$.
To argue for a second state in the list, we employ the probabilistic method once more.
By Equation~\eqref{eq:far_away_from_fixed} resp. \eqref{eq:far_away_from_fixed_designs} the probability that a pure state is close to any fixed pure state $\ketbra{\phi}$, $ \mathrm{Pr} \left[ \tfrac{1}{2} \left\| \ketbra{\psi} -\ketbra{\phi} \right\|_1 \leq 1-\Delta \right]$, which we denote for all cases by $p_G$, is exponentially (resp. inverse polynomially) suppressed in $D$.
By the union bound, the probability of picking a low-complexity state that is close to the previously chosen high-complexity state $\ket{\psi_1}$ is \begin{align}
    \mathrm{Pr} \left[ \mathcal{C}_\delta (| \psi \rangle) \leq r \cup \frac{1}{2} \norm{\ketbra{\psi} - \ketbra{\psi_1}}_1 \leq 1-\Delta\right]
    &\leq \mathrm{Pr} \left[ \mathcal{C}_\delta (| \psi \rangle) \leq r\right]  + \mathrm{Pr} \left[\frac{1}{2} \norm{\ketbra{\psi} - \ketbra{\psi_1}}_1 \leq 1-\Delta\right] \nonumber\\
    &\leq \frac{1}{2} +p_G.
\end{align}
Hence, the complementary probability is strictly greater than $ \frac{1}{2}-p_G$, and there exists a high complexity state $\ket{\psi_2}$ that is almost orthogonal to the previous choice $\ket{\psi_1}$. Repeating this leads to an ensemble of high-complexity states that are almost maximally far apart until the number of states $N$ does not counterbalance the suppression in $D$ of the probability $p_G$.
This is the case until $Np_G > \frac{1}{2}$.
\end{proof}
\subsection{Proof of the uniformity of the average Born distribution}
\begin{lemma}\label{lemma:uniform}
Then the average Born distribution induced by the classical compact groups equals the uniform distribution, i.e.
\begin{align}
    \mathbb{E}_{U\sim\mu_G}\,P_U=\mathcal{U}.
\end{align}
\end{lemma}
\begin{proof}
We will use the fact that all classical compact groups are unitary 1-designs.
Given that
$\ketbra{0}^{\otimes n} = \frac{1}{D} (\mathbb{I}+ Z)^{\otimes n}$,
\begin{align}
    \mathbb{E}\ketbra{\psi}  &= \mathbb{E}_{U\sim \mu_G} U \ketbra{0}^{\otimes n} U^\dag  = \mathcal{E}^{(1)}(\ketbra{0}^{\otimes n}) =
    \sum_{\eta = 1}^{\dim \text{comm}(G)} \Tr(B_{\eta} \ketbra{0}^{\otimes n})B_{\eta}
    \\&= \frac{1}{D} \sum_{\eta} \left(\Tr({B}_{\eta}) +\sum_i \Tr({B}_{\eta} Z_i) + \sum_{i< j}\Tr( {B}_{\eta} Z_i Z_j) + \cdots \Tr({B}_{ \eta} Z^{\otimes n})  \right) {B}_{\eta}
    \\&= \frac{1}{D}\left( \mathbb{I} + \sum_{\eta} \left(\sum_i \delta_{{B}_{\eta} Z_i} Z_i + \sum_{i<j} \delta_{{B}_{\eta} Z_{ij}} Z_{ij} + \cdots + \delta_{{B}_{\eta} Z^{\otimes n}} Z^{\otimes n}
    \right)\right).
\end{align}
where in the last equality, we used that all Pauli operators are traceless except the identity.
Therefore, the expectation is determined by the intersection of basis elements of the commutant with $ \left\{\frac{\mathbb{I}}{\sqrt{2}}, \frac{Z}{\sqrt{2}} \right\}^{\otimes n}$. Given that the commutant consists solely of the identity, the average state is the maximally mixed state $\mathbb{E}\ketbra{\psi}  = \frac{\mathbb{I}_D}{D}$.
Now,
\begin{align}\label{eq:1mom}
    \mathbb{E}\, P_U(x) = \langle \mathbb{E}\ketbra{\psi}, \ketbra{x} \rangle = \frac{1}{D}  \eqcolon \mathcal{U}(x).
\end{align}
Hence, the average Born distribution is the uniform distribution.
\end{proof}
\subsection{Proof of Theorem~\ref{th:Gaussian_SO}}\label{sec:proof_Gaussian_SO}
\begin{theoremrestated}{th:Gaussian_SO}
Let $f : \mathbb{R}^D \to \mathbb{R}$ be homogeneous of even degree $2k$ and let $\mu_{SO}$ denote the Haar-induced measure on real pure states.
Then
\begin{align}
\mathbb{E}_{\ket{\psi}\sim \mu_{SO}}[ f(\ket{\psi}) ]
=
\frac{\mathbb{E}_{g\sim \mathcal{N}(0,\mathbb{I}_D)}[f(g)]}{k!\,2^k\binom{D/2+k-1}{k}}.
\end{align}
\end{theoremrestated}
\begin{proof}
First, we note that a Gaussian random vector $X \sim \mathcal{N}(0,\mathbb{I}_D)$ can be represented in the polar form  as $X = r \theta$
where $r = \norm{X}_2 $ is the length and $\theta = \frac{X}{r}$ the direction of $X$.
$r$ and $\theta$ are independent.
Similarly as in Ref.~\cite[Theorem 17]{nietner2023average}, for a homogeneous function with an even degree $2k \in 2 \mathbb{N}$, we can therefore simplify the expectation value as follows:
\begin{align}\label{eq:Gaussian}
     \mathbb{E}_{g \sim \mathcal{N}(0,\mathbb{I}_D)} f(g) &=
     \mathbb{E}_{r \sim \chi_{D}, \theta \sim \mu_S} f( r \theta ) =   \mathbb{E}_{r \sim \chi_{D}, \theta \sim \mu_S} r^{2k} f(\theta)
   \\& =  \mathbb{E}_{r^2 \sim \chi^2_{D}} r^{2k}    \mathbb{E}_{ \theta \sim \mu_S} f(\theta).
\end{align}
The second equality uses the assumption that $f$ is a homogeneous function. The third equality follows from the independence of $r$ and $\theta$.
$\mu_S$ denotes the uniform distribution on the real unit sphere.
By~\cite[Theorem 3.7]{mattila1999geometry}
 pure Haar random states over the special orthogonal group are distributed uniformly on the real sphere.
The moments of $\chi^2$ distributed random variables with $D$ degrees of freedom are
\begin{align}\label{eq:chi2moments}
\mathbb{E}_{r^2 \sim \chi^2_{D}}\, r^{2k}
= 2^{k}\,\frac{\Gamma\!\left(k+\frac{D}{2}\right)}{\Gamma\!\left(\frac{D}{2}\right)}
= k!\,2^k\,\binom{\frac{D}{2}+k-1}{k}.
\end{align}
Rearranging yields the claimed identity.
\end{proof}
\subsection{Bounds on the total variation distance between the Born and uniform distribution} \label{sec:proof_constants}
Now we can use Theorem \ref{th:Gaussian_SO} to bound the total variation distance between the Born and uniform distribution.
\begin{lemma}\label{lem:tv_distance_bounds}
Let $U\sim \mu_{SO}$ and let $P_U$ be the induced Born distribution. Define
\[
M_{SO}\coloneqq \frac12\mathbb{E}_{z\sim \mathcal{N}(0,1)}\bigl|z^2-1\bigr|=\sqrt{\frac{2}{\pi e}},
\qquad
\Delta_{SO}\coloneqq \frac12\mathbb{E}_{\rho\sim \chi^2_D}\left|1-\frac{\rho}{D}\right|.
\]
Then $0\le \Delta_{SO}\le \frac{1}{\sqrt{2D}}$ and
\[
M_{SO}-\Delta_{SO}\le \mathbb{E}_{U} d_{TV}(P_U,\mathcal{U}) \le M_{SO}+\Delta_{SO}.
\]
For $SU(D)$ (and likewise $Sp(D/2)$), the same inequality holds with $M_{SU}=1/e$ and $\Delta_{SU}\leq \frac{1}{2\sqrt{D}}$ as in \cite[Theorem 19]{nietner2023average}.
\end{lemma}

\begin{proof}

Describing the total variation distance between the Born distribution and the uniform distribution as a homogeneous function with an even degree will allow us to compute its expectation value.
As~\cite[Section 3.2]{nietner2023average},
we recall
\begin{align}
   & d_{TV}(P_U, \mathcal{U} ) =
      \frac{1}{2}  \sum_x \abs{P_U(x) - \mathcal{U}(x)}  =  \frac{1}{2} \sum_x \abs{\abs{\bra{x} U \ket{0}}^2 - \frac{1}{D} } =\frac{1}{2D} \sum_x \abs{D\abs{\bra{x} U \ket{0}}^2 - 1 }.
\end{align}
Hence,
\begin{align}
     \mathbb{E}_{U \sim \mu_U} d_{TV}(P_U, \mathcal{U} ) &=
      \frac{1}{2D}  \sum_x \mathbb{E}_{\ket{\psi} \sim \mu}   \abs{ D \abs{\bra{x}\ket{\psi}}^2 - 1}.
\end{align}
Also, $f_x(\ket{\psi}) :=  \abs{D \abs{\braket{x}{\psi}}^2- \norm{\ket{\psi}}_2^2} $ is a homogeneous function of degree two as for $ a\in \mathbb{R}$:
\begin{align}\label{eq:fx}
  f_x(a\ket{\psi}) &=  \abs{D \abs{\braket{x}{a \psi}}^2- \norm{a\ket{\psi}}_2^2}
    = a^2 f_x(\ket{\psi}).
\end{align}
Orthogonal matrices have real entries and the computational basis vectors are real-valued. Thus, the complex inner product reduces to the real inner product, and the domain of $f_x$ is $\mathbb{R}^D$. We can therefore apply Gaussian integration to determine its expectation value:
\begin{align}
  \mathbb{E}_{\ket{\psi} \sim \mu}  f_x(\ket{\psi}) &= \frac{1}{D}  \mathbb{E}_{g \sim \mathcal{N}(0,\mathbb{I})} f_x(g) =  \frac{1}{D} \mathbb{E}_{g \sim \mathcal{N}(0,\mathbb{I})} \abs{D \abs{x \cdot g}^2- \norm{g}_2^2}.
\end{align}
The inner product of the computational basis vector with the Gaussian vector $x \cdot g$ picks one component of the Gaussian vector, yet, $\forall i \in \{1, \dots, D\},\;  g_i \sim \mathcal{N}(0,1) $. Therefore, we can pick w.l.o.g. the first component:
\begin{align}
    \mathbb{E}_{\ket{\psi} \sim \mu}  f_x(\ket{\psi}) &=\frac{1}{D} \mathbb{E}_{g_i \overset{iid}{\sim} \mathcal{N}(0,1)} \abs{D g_1^2-  \sum_{i=1}^{D} g_i^2}.
\end{align}
Hence, the expectation value of $f_x$ is independent of the computational basis state, and in total, the expectation of the total variation distance between the Born distribution and the uniform distribution simplifies to
\begin{align}
     \mathbb{E}_{U \sim \mu_U} d_{TV}(P_U, \mathcal{U} ) &= \frac{1}{2} \mathbb{E}_{g_i \overset{iid}{\sim} \mathcal{N}(0,1)} \abs{g_1^2 - \frac{1}{D} \sum_{i=1}^D g_i^2}.
\end{align}
Through the triangle inequality and the reverse triangle inequality, we can bound this expectation value in terms of the asymptotic mean value $M_{SO} = \frac{1}{2}\mathbb{E}_{g_1 \sim \mathcal{N}(0,1)} \abs{ g_1^2- 1}$ and  $\Delta_{SO} = \frac{1}{2} \mathbb{E}_{g_i \overset{iid}{\sim} \mathcal{N}(0,1)} \abs{1- \frac{1}{D}  \sum_{i=1}^{D} g_i^2} =  \frac{1}{2} \mathbb{E}_{\rho \sim \chi^2} \abs{1- \frac{1}{D} \rho}$.
We get,
\begin{align}
    M_{SO} &= \frac{1}{2} \frac{1}{\sqrt{2\pi}} \int_{-\infty}^{\infty} \abs{x^2-1} e^{-x^2/2} dx \\
    &=\frac{1}{\sqrt{2\pi}} \int_{0}^{\infty}  \abs{x^2-1}e^{-\frac{x^2}{2}} dx
    \\&= \frac{1}{\sqrt{2\pi}} \left(\int_{0}^{1} (1-x^2)e^{-\frac{x^2}{2}} dx+  \int_1^{\infty} (x^2-1) e^{-\frac{x^2}{2}} dx\right) \\&
= \frac{1}{\sqrt{2\pi}}\left(
e^{-x^2/2}x\Big|_0^1 - e^{-x^2/2}x\Big|_1^\infty
\right)
=\sqrt{\frac{2}{ \pi e}}.
\end{align}
In the second line, we used that the integrand is even. In the third line, we split the integration area according to the definition of the absolute value.
In the last line, we used the Gaussian integral
\begin{align}
&\int e^{-x^2/2} (1-x^2) dx =  e^{-x^2/2} x + c.
\end{align}

We bound $\Delta_{SO}$ via Jensen's inequality for the concave square root and using that the variance of a $\chi^2$-distributed random variable with $D$ degrees of freedom is $2D$:
\begin{align}
 \Delta_{SO} = \frac{1}{2D}  \mathbb{E}_{\rho \sim \chi^2} \abs{D-  \rho} \leq  \frac{1}{2D} \left(\mathbb{E}_{\rho \sim \chi^2} (D-\rho)^2\right)^{\frac{1}{2}} = \frac{1}{ \sqrt{2D}}.
\end{align}

In total, we can bound the expected total variation distance between the Born distribution and target distribution as follows:
\begin{align}
  M_{SO}- \Delta_{SO} \leq \mathbb{E}_{U \sim \mu_U} d_{TV}(P_U, \mathcal{U} ) \leq  M_{SO}+ \Delta_{SO}.
\end{align}
In the limit of large system size, $\Delta_{SO}$ vanishes.

For the unitary group the same bound with $M_{SU} = \frac{1}{e} , \Delta_{SU}= 2^{- \frac{n}{2} -1}$ by~\cite[Theorem 19]{nietner2023average}. The bound with the same constants as in the unitary case holds for $\mathcal{P}_{Sp}$ as for pure states  $\mathcal{S}_{Sp}$ is an exact unitary state design for all $k$ by~\cite[Theorem 1]{west2024randomensemblessymplecticunitary}.

\end{proof}

\subsection{Proof of Theorem~\ref{th:average_case_hardness}}\label{sec:proof_hardness}

\begin{theoremrestated}{th:average_case_hardness}
Fix $\varepsilon,\tau\ge 0$ and set $\xi_G \coloneqq M_G-\Delta_G-(\varepsilon+\tau)\ge 0$, where $(M_G,\Delta_G)$ are defined as above.
Any algorithm that $\varepsilon$-learns a $\beta$-fraction of Born distributions $\mathcal{P}_G$ with $\tau$-accurate statistical queries must use $q$ queries satisfying
\begin{align*}
q+1 \;\ge\;
\begin{cases}
\displaystyle
\frac{\beta - 2 \exp\!\left(- \frac{(2^{n}-2)\,\xi_{SO}^2}{8}\right)}
     {2 \exp\!\left(-\frac{ (2^n -2)\,\tau^2}{32}\right)}
&\text{for } G=SO(D),
\\[1.25em]
\displaystyle
\frac{\beta - 2 \exp\!\left(- \frac{(2^{n-1}+1)\,\xi_{SU}^2}{2}\right)}
     {2 \exp\!\left(-\frac{ (2^{n-1} +1)\,\tau^2}{8}\right)}
&\text{for } G=Sp(D/2),
\\[1.25em]
\displaystyle
\frac{\beta - 2 \exp\!\left(- \frac{2^n \,\xi_{SU}^2 }{4}\right)}
     {2 \exp\!\left(- \frac{2^n \,\tau^2 }{16}\right)}
&\text{for } G=SU(D).
\end{cases}
\end{align*}
\end{theoremrestated}
\begin{proof}
The maximally distinguishable fraction $\mathfrak{f}$ can be bounded for all classical compact groups by Lévy's Lemma, given that
$ f\colon G \to \mathbb{R}, U \mapsto  \mathbb{E}_{x \sim P_U}\phi(x)  = \text{Tr}(\ketbra{0}^{\otimes n} U^\dag \Phi U )$,
$
    \Phi = \sum_x \phi(x) \ketbra{x}
$, is $2$-Lipschitz combining~\cite[Lemma 22]{nietner2023average} and Lemma \ref{lemma:lipschitz}. Moreover, the average over all pure states yields the maximally mixed state  by Lemma \ref{lemma:uniform} and by linearity  $\mathbb{E}_{U\sim \mu_G} f(U) = \mathbb{E}_{x \sim \mathcal{U}} \phi(x) $.
Hence $\mathfrak{f} \leq  2 e^{-\frac{\tau^2 }{8C_G}} $.
The total variation distance between the Born and uniform distributions is $1$-Lipschitz by~\cite[Lemma 23]{nietner2023average} together with Lemma~\ref{lemma:lipschitz}.
Lemma~\ref{lem:tv_distance_bounds} bounds the average total variation distance to the uniform distribution.
Setting  $\xi_G = M_G - \Delta_G -(\varepsilon + \tau) \geq 0 $, the probability of trivial learning can be bounded as $\mathfrak{u}\leq 2 e^{-\frac{\xi_G^2}{2C_G}}$. The query complexity lower bound follows then from Lemma \ref{lemma:nie}.
\end{proof}
\end{document}